\documentclass[10pt,twocolumn]{article}

\usepackage[T1]{fontenc}
\usepackage[utf8]{inputenc}
\usepackage{lmodern}
\usepackage{microtype}
\usepackage{titlesec}
\usepackage{etoolbox}
\usepackage[font=small,labelfont=bf]{caption}

\usepackage{amsmath,amssymb,amsthm,mathtools}
\usepackage{graphicx}
\usepackage{subcaption}
\usepackage{booktabs}

\usepackage{hyperref}
\hypersetup{
    colorlinks=true,
    linkcolor=black,
    citecolor=black,
    urlcolor=black
}

\usepackage{cleveref}
\usepackage{geometry}
\newtheorem{theorem}{Theorem}
\newtheorem{lemma}[theorem]{Lemma}

\newtheorem{condition}[theorem]{Condition}
\newtheorem{definition}[theorem]{Definition}
\newtheorem{remark}[theorem]{Remark}

\renewcommand{\qedsymbol}{$\blacksquare$}
\newcommand{\placeqed}{\nobreak\hfill \ensuremath{\blacksquare}}

\usepackage{tikz}
\newcommand{\blackdashedline}{\raisebox{2pt}{\tikz{\draw[-,black!40!black,dashed,line width = 0.9pt](0,0) -- (5mm,0);}}}
\newcommand{\blackline}{\raisebox{2pt}{\tikz{\draw[-,black!40!black,line width = 0.9pt](0,0) -- (5mm,0);}}}

\newcommand{\reddashedline}{\raisebox{2pt}{\tikz{\draw[-,black!40!red,dashed,line width = 0.9pt](0,0) -- (5mm,0);}}}

\newcommand{\blueline}{\raisebox{2pt}{\tikz{\draw[-,black!40!blue,line width = 0.9pt](0,0) -- (5mm,0);}}}
\newcommand{\orangedottedline}{\raisebox{2pt}{\tikz{\draw[-,orange!40!orange,dotted,line width = 0.9pt](0,0) -- (5mm,0);}}}
\definecolor{mypurple}{HTML}{7B3294}
\newcommand{\purpleline}{\raisebox{2pt}{\tikz{\draw[-,mypurple!40!mypurple,line width = 0.9pt](0,0) -- (5mm,0);}}}

\titleformat{\section}
  {\normalfont\bfseries\normalsize}
  {\thesection}{0.6em}{}

\titleformat{\subsection}
  {\normalfont\bfseries\small}
  {\thesubsection}{0.6em}{}

\titleformat{\subsubsection}
  {\normalfont\bfseries\footnotesize}
  {\thesubsubsection}{0.6em}{}

\titlespacing*{\section}
  {0pt}{1.5ex plus 0.5ex minus 0.25ex}{0.75ex}

\titlespacing*{\subsection}
  {0pt}{1.1ex plus 0.4ex minus 0.2ex}{0.6ex}

\titlespacing*{\subsubsection}
  {0pt}{0.9ex plus 0.2ex minus 0.1ex}{0.5ex}

\newcommand\blfootnote[1]{%
  \begingroup
  \renewcommand\thefootnote{}%
  \footnote{$^*$#1}%
  \addtocounter{footnote}{-1}%
  \endgroup
}

\makeatletter
\renewenvironment{thebibliography}[1]
     {\section*{\refname}%
      \@mkboth{\MakeUppercase\refname}{\MakeUppercase\refname}%
      \list{\@biblabel{\@arabic\c@enumiv}}%
           {\settowidth\labelwidth{\@biblabel{#1}}%
            \setlength{\itemsep}{2pt} 
            \setlength{\parsep}{2pt}  
            \setlength{\leftmargin}{\labelwidth}
            \addtolength{\leftmargin}{\labelsep}
            \usecounter{enumiv}%
            \let\p@enumiv\@empty
            \renewcommand\theenumiv{\@arabic\c@enumiv}}%
      \sloppy
      \clubpenalty4000
      \widowpenalty4000
      \sfcode`\.=1000\relax}
     {\def\@noitemerr{\@latex@warning{Empty `thebibliography' environment}}%
      \endlist}
\makeatother
\apptocmd{\thebibliography}{\footnotesize}{}{}

\title{\textbf{Efficient stochastic model-predictive control based on the meta-state-space representation}$^*$\\(extended version)}

\author{
    Bendegúz Györök\textsuperscript{1}, Roland Tóth\textsuperscript{1,2}, Maarten Schoukens\textsuperscript{2}, Tamás Péni\textsuperscript{1}\\
    \small\textsuperscript{1}\textit{Systems and Control Laboratory, HUN-REN Institute for Computer Science and Control, Budapest, Hungary}\\
    \small\textsuperscript{2}\textit{Control Systems Group, Eindhoven University of Technology, Eindhoven, The Netherlands}\\
    \small Emails: \texttt{gyorokbende@sztaki.hu, r.toth@tue.nl, m.schoukens@tue.nl, peni@sztaki.hu}
}

\date{} 

\begin{document}
\twocolumn[
\maketitle
\begin{abstract}
\small
\emph{Stochastic model-predictive control} (SMPC) has evolved to a powerful framework for the control of stochastic dynamical systems. SMPC utilizes a probabilistic uncertainty description to provide a systematic trade-off between the control objective and constraint satisfaction in a statistical sense. However, the majority of existing SMPC methods face challenges related to computational tractability due to the need for stochastic inference. Approaches that apply accurate inference are computationally demanding, which can lead to serious limitations in the implementability of these methods. Hence, in practice, the uncertainty propagation and the resulting distributions are typically approximated, e.g., by Gaussian distributions. These approximations promote computational efficiency, but are often too conservative, becoming a limiting factor in the representation of stochastic state evolution and the implied guarantees. To overcome this fundamental limitation of SMPC approaches, we propose a novel formulation based on the \emph{meta-state-space} (MSS) representation of stochastic dynamical systems. The proposed MSS-based SMPC scheme offers a computationally efficient way to forward propagate the uncertainty with a flexible and highly accurate approximation of the probabilistic system description. With the presented method, the control objective can be formulated in terms of the output probability density function, while the resulting optimization problem remains compatible with real-time operation, even for general nonlinear stochastic systems. Finally, we provide a detailed theoretical analysis and demonstrate the effectiveness of the proposed methodology via an extensive simulation study.
\end{abstract}
\vspace{0.5em}
\vspace{1.5em}
]
\blfootnote{This project has been supported by the Air Force Office of Scientific Research under award number FA8655-23-1-7061 and by The MathWorks Inc. Corresponding author: B.M.~Györök.}

\section{Introduction}\label{sec:intro}
\emph{Model-predictive control} (MPC) is a widely used optimization-based control technique that has achieved success mainly because it guarantees optimality with respect to a general objective function while also ensuring various types of constraints. \emph{Stochastic model-predictive control} (SMPC) can be interpreted as a probabilistic version of classical MPC by modelling disturbances, model mismatch, and measurement noise as stochastic processes. In contrast to nominal MPC formulations, SMPC incorporates probabilistic descriptions of uncertainty directly into the prediction model and enforces constraints in a statistical sense with a prescribed probability. For a comprehensive overview of SMPC approaches, see~\cite{mesbah_stochastic_2016}. In practice, SMPC methods have been successfully utilized in various application fields, e.g., for process control \cite{van_hessem_stochastic_2006}, in epidemic management~\cite{herceg_scenario-based_2025}, in the automotive industry \cite{liu_stochastic_2015}, etc.

Early SMPC developments primarily focused on linear systems with additive stochastic disturbances, where uncertainty propagation and chance constraint evaluation admit tractable analytic reformulations \cite{schwarm_chance-constrained_1999}. Such approaches preserve computational efficiency, but rely heavily on Gaussian assumptions and linear dynamics. For nonlinear systems, a common strategy is to approximate the propagated uncertainty with Gaussian distributions, obtained, e.g., through linearization or moment-matching techniques. This modelling choice is largely motivated by computational tractability, since employing more complex probability distributions to describe state evolution quickly leads to intractable inference. However, while promoting computational efficiency, these approximations can result in conservative predictions and inaccurate probability estimates, particularly in the presence of strong nonlinearities or non-Gaussian disturbances \cite{mesbah_stochastic_2016}. Similar approximations are frequently employed for chance constraint evaluation, where exact computation of constraint satisfaction probabilities is generally intractable.

More recently, increasing attention has been given to SMPC formulations based on \emph{Gaussian process} (GP) models, often referred to as GP-MPC. In these approaches, Gaussian processes are employed to model unknown (or partially known) system dynamics in a data-driven manner. The GP formulation provides both a mean prediction and an explicit quantification of model uncertainty, which naturally integrates well into the SMPC framework \cite{kocijan_gaussian_2004}. GP-MPC schemes have been successfully applied to nonlinear systems, and several works have investigated stability, constraint satisfaction, and safe learning properties \cite{hewing_cautious_2020}. However, exact uncertainty propagation through GP models is generally intractable and practical implementations rely on the same approximations (linearization, moment matching), which may again introduce conservatism into the uncertainty characterization.

Besides GP-based formulations, several works represent or directly control the probability distribution itself. Output pdf control approaches approximate the output density using B-spline expansions with input-dependent weight dynamics~\cite{wang_robust_1999}; however, they assume that the output pdf can be measured online, and the B-spline approximation at each time step remains computationally expensive. Fully probabilistic feedback design minimizes the Kullback-Leibler divergence to a desired closed-loop distribution~\cite{karny_towards_1996}, but cannot handle constraints or provide stability guarantees and suffers from the curse of dimensionality. SMPC based on \emph{polynomial chaos} (PC) theory provides a finite-dimensional representation of parametric uncertainties~\cite{mesbah_stochastic_2014}, but requires online computation of the PC coefficients for each control input, with the number of terms growing rapidly with the number of uncertain parameters. Fokker-Planck-based SMPC~\cite{buehler_lyapunov-based_2016} enables state-pdf shaping for control-affine systems with state-dependent additive disturbances, but requires solving a partial differential equation online, limiting scalability with the state dimension.

To overcome fundamental challenges of SMPC methods related to computational tractability, uncertainty propagation in nonlinear systems, and accurate yet efficient evaluation of chance constraints, we adopt the recently proposed \emph{meta-state-space}~(MSS) representation of stochastic dynamical systems~\cite{beintema_meta-statespace_2024} and develop an MSS-based SMPC scheme. The employed modeling approach transforms the original stochastic state-space into a \emph{deterministic} meta-state-space, with all uncertainty captured through the model output probability function. Due to the deterministic nature of the (generally nonlinear) meta-state transition function, uncertainty propagation is no longer required during the solution of the MPC optimization problem. This solves the computational tractability issues typically associated with SMPC methods. After forward simulation of the deterministic meta-states, chance constraints can be evaluated directly. Consequently, the proposed chance constraint evaluation method requires minimal approximations in contrast to existing SMPC techniques.

Beyond these benefits, the proposed MSS-based SMPC formulation holds two other key advantages over conventional SMPC approaches. First, the control objective can be formulated directly at the output distribution level, enabling the specification of a desired output probability density function. This is in contrast with most existing methods, which formulate the control objective for the expected value, occasionally augmented with variance-related terms. Second, meta-state-space models are typically parametrized and identified using neural-network-based representations. During model learning, a so-called encoder network can be co-estimated with the model parameters to estimate the initial meta-state values from past input-output samples. Since state estimation is a well-known challenge in SMPC (see, e.g., \cite{cannon_stochastic_2012}), the proposed approach addresses this issue inherently by utilizing the encoder in the SMPC formulation, acting as a meta-state observer.

The contributions of this paper can be summarized as:
\begin{itemize}
    \item [C1] Formulating an SMPC approach based on the meta-state-space formulation with cost function options, terminal ingredients, and chance constraint evaluation, setting up all the necessary aspects for practical implementations.
    \item [C2] Proving recursive feasibility and asymptotic stability of the proposed meta-state-space-based SMPC scheme by extending the classical results of~\cite{rawlings_model_2017}.
    \item [C3] Extending the proposed approach to handle online-changing set-points, building on the arguments of \cite{soloperto_nonlinear_2023}.
    \item [C4] Demonstrating the effectiveness of the proposed SMPC formulation with numerical examples.
\end{itemize}

The paper is organized as follows: Section~\ref{sec:MSS_approach} provides a brief description of the meta-state-space representation of stochastic dynamical systems. Then, Section~\ref{sec:MSS-MPC} introduces the proposed SMPC formulation based on the meta-state-space form, specifically tailored for a constant control objective. After the SMPC formulation, recursive feasibility and asymptotic stability are proven; thus, Section~\ref{sec:MSS-MPC} covers Contributions 1 and 2 of the paper. Then, Section~\ref{sec:changing_ref} extends the methodology to handle online-changing control objectives, i.e., introduces Contribution 3. Section~\ref{sec:num_examples} demonstrates the effectiveness of the proposed approach via a rigorous numerical case study. Finally, conclusions are drawn in Section~\ref{sec:conclusion}.

\subsection{Notation}
We denote the Gaussian distribution with mean $\mu$ and variance $\sigma$, as $\mathcal{N}(\mu, \sigma)$, while, $\mathcal{U}(a, b)$ denotes the uniform distribution with support from $a$ to $b$, as $\mathcal{U}(a,b)\in\left[a,b\right]$. By $x_{a:b}$, we denote the sequence $\{x_k\}_{k=a}^b$, $a, b\in \mathbb{Z}$ with $b\geq a$. The set of integers in the closed interval $[a,b]$ is denoted by $\mathbb{I}_a^b$. Furthermore, $\mathcal{K}_\infty$ denotes the set of functions $\alpha:\mathbb{R}_{\geq 0} \rightarrow \mathbb{R}_{\geq 0}$ that are continuous, strictly increasing, and unbounded with $\alpha(0)=0$. We denote the quadratic norm w.r.t. a symmetric, positive definite matrix $A$, i.e., $A\succ 0$, by $\|x\|_A^2 = x^\top A x$. The Euclidean vector norm, i.e., the 2-norm, is denoted as $\|\cdot\|_2$. Let $\mathbb{K}^{n}$ denote the set of symmetric, positive-definite matrices in $\mathbb{R}^{n\times n}$. We denote the minimal and maximal eigenvalue of a symmetric matrix $A=A^\top$ as $\lambda_\mathrm{min}(A)$, and $\lambda_\mathrm{max}(A)$, respectively. Let $\mathcal{C}_2$ denote the set of continuous functions whose first and second derivatives exist and are continuous. Let 
$\mathcal{P}(\mathbb{R}^{n})$ be the set of all density functions $p: \mathbb{R}^{n} \rightarrow \mathbb{R}$ that are non-negative and Lebesgue integrable. Note that with the Borel $\sigma$-algebra over $\mathbb{R}^{n}$, these functions are Lebesgue integrable over all Borel sets of $\mathbb{R}^{n}$.

\section{Meta-state-space representation}\label{sec:MSS_approach}
The dynamics of the stochastic system to be controlled are considered to be defined by a \emph{discrete-time} nonlinear SS representation:
\begin{equation}
\label{eq:data_gen_sys}
    \mathbf{x}_{k+1}=f(\mathbf{x}_k, u_k, \mathbf{v}_k),\quad
    \mathbf{y}_k=h(\mathbf{x}_k, u_k, \mathbf{e}_k),
\end{equation}
where $k\in\mathbb{Z}$ is the discrete time index, $\mathbf{x}_k$ random variable represents the state, taking values in $\mathbb{R}^{n_\mathrm{x}}$, $u_k\in\mathbb{R}^{n_\mathrm{u}}$ is the deterministic control input signal, and $\mathbf{y}_k$ represents the measured output, which is a random variable taking values from $\mathbb{R}^{n_\mathrm{y}}$. The state transition is influenced by an i.i.d. stationary process noise $\mathbf{v}_k$ with \emph{probability density function} (pdf) $p^\mathbf{v}\in\mathcal{P}(\mathbb{R}^{n_\mathrm{v}})$, and the measured output is corrupted by some i.i.d. stationary measurement noise $\mathbf{e}_k$ with pdf $p^\mathbf{e} \in \mathcal{P}(\mathbb{R}^{n_\mathrm{e}})$. Moreover, the state transition $f:\mathbb{R}^{n_\mathrm{x}} \times \mathbb{R}^{n_\mathrm{u}}\times \mathbb{R}^{n_\mathrm{v}} \rightarrow \mathbb{R}^{n_\mathrm{x}}$, and the output map $h:\mathbb{R}^{n_\mathrm{x}} \times \mathbb{R}^{n_\mathrm{u}}\times \mathbb{R}^{n_\mathrm{e}}\rightarrow \mathbb{R}^{n_\mathrm{y}}$ are both (possibly) nonlinear functions.

For ease of notation, the state and output pdfs at time $k$ are indicated as $p_k^\mathbf{x}$ and $p_k^\mathbf{y}$, respectively. Using this notation, all future state and output probability distributions can be expressed, given an initial state distribution $p_0^\mathbf{x}\in\mathcal{S}_\mathbf{x}^0$ and a known input sequence $u_{0:\infty}$. Here, $\mathcal{S}_\mathbf{x}^0$ denotes the set of all possible initial state distributions with $\mathcal{S}_\mathbf{x}^0\subseteq \mathcal{P}(\mathbb{R}^{n_\mathrm{x}})$. Therefore,
\begin{equation}\label{eq:chapman_kolg}
    p_{k+1}^\mathbf{x} = F(p_k^\mathbf{x}, u_k),\quad
    p_{k}^\mathbf{y} = H(p_k^\mathbf{x}, u_k),
\end{equation}
where $F:\mathcal{P}_\mathbf{x} \times \mathbb{R}^{n_\mathrm{u}}\rightarrow \mathcal{P}_\mathbf{x}$ and $H:\mathcal{P}_\mathbf{x} \times \mathbb{R}^{n_\mathrm{u}}\rightarrow \mathcal{P}_\mathbf{y}$ are the \emph{Chapman-Kolmogorov} equations~\cite{paul_stochastic_2013} with $\mathcal{P}_\mathbf{x}\triangleq \mathcal{P}(\mathbb{R}^{n_\mathrm{x}})$ and $\mathcal{P}_\mathbf{y}\triangleq \mathcal{P}(\mathbb{R}^{n_\mathrm{y}})$. In \eqref{eq:chapman_kolg}, $F$ recursively propagates the state distribution in time, starting from $p_0^\mathbf{x} \in \mathcal{S}_\mathbf{x}^0$, while $H$ provides the output distributions. We can also characterize the set of all possible state distributions that occur along any trajectory of \eqref{eq:data_gen_sys}, as
\begin{equation}
    \mathcal{S}_\mathbf{x}\triangleq \{p_k^\mathbf{x}\mid p_{k+1}^\mathbf{x} = F(p_k^\mathbf{x}, u_k), p_0^\mathbf{x} \in \mathcal{S}_\mathbf{x}^0,
    u_{0:\infty}\in\mathbb{U}, k\geq 0\},
\end{equation}
where $\mathbb{U}$ is the set of all possible (deterministic) input trajectories, i.e., $\mathbb{U} = \{\{u_k\}_{k=0}^\infty \mid u_k\in\mathbb{R}^{n_\mathrm{u}}, k\geq 0\}$. As shown in~\cite{beintema_meta-statespace_2024}, if $\mathcal{S}_\mathbf{x}$ can be uniquely parametrized of order~$n_\mathrm{z}$, then there exists an equivalent representation of \eqref{eq:data_gen_sys}, in the so-called \emph{meta-state-space} (MSS) form:
\begin{equation}\label{eqs:meta-system}
    z_{k+1} = f^\mathrm{z}(z_k, u_k),\quad
    p_k^\mathbf{y} = h^\mathrm{z}(z_k, u_k),
\end{equation}
where $z_k\in\mathbb{R}^{n_\mathrm{z}}$ is the meta-state, $f^\mathrm{z}:\mathbb{R}^{n_\mathrm{z}}\times \mathbb{R}^{n_\mathrm{u}}\rightarrow \mathbb{R}^{n_\mathrm{z}}$ is the meta-state-transition function, while $h^\mathrm{z}:\mathbb{R}^{n_\mathrm{z}}\times \mathbb{R}^{n_\mathrm{u}}\rightarrow \mathcal{P}_\mathbf{y}$ is the output probability function. For a detailed derivation, proof, and conditions of existence, see~\cite{beintema_meta-statespace_2024}. For certain systems, such as linear time-invariant systems with Gaussian disturbances, analytic conversion of \eqref{eq:data_gen_sys} to \eqref{eqs:meta-system} is directly available. However, for more complex systems, analytic formulas for conversion to \eqref{eqs:meta-system} are not known. Instead, it is directly estimated from input-output data, generated by \eqref{eq:data_gen_sys}, in the form of
\begin{equation}\label{eqs:meta-ANN-model}
    \hat{z}_{k+1} = f^\mathrm{z}_\theta(\hat{z}_k, u_k),\quad
    \hat{p}_k^\mathbf{y} = h_\theta^\mathrm{z}(\hat{z}_k, u_k),
\end{equation}
where $\hat{z}_k\in\mathbb{R}^{n_{\hat{\mathrm{z}}}}$ is the model meta-state, $\hat{p}_k^\mathbf{y}$ is the predicted output pdf at time $k$, and $f^\mathrm{z}_\theta$, $h^\mathrm{z}_\theta$ are both parametrized by $\theta\in\mathbb{R}^{n_\theta}$. Typically,  $f^\mathrm{z}_\theta$ is parametrized as a fully-connected feedforward \emph{artificial neural network} (ANN), while $h^\mathrm{z}_\theta$ is described by the sum of $n_\mathrm{G}$ Gaussian components, as
\begin{equation}\label{eq:h_theta}
    h_\theta^\mathrm{z}(\xi_k) = \sum_{i=1}^{n_\mathrm{G}} w^{(i)}_\theta (\xi_k) \cdot \mathcal{N}\left(\mu^{(i)}_\theta(\xi_k), \Sigma^{(i)}_\theta(\xi_k)\right),
\end{equation}
where $\xi_k = [\hat{z}_k^\top\, u_k^\top]^\top$ is introduced for notational convenience, $w_\theta^{(i)} : \mathbb{R}^{n_\mathrm{z}+n_\mathrm{u}} \rightarrow \mathbb{R}$, $\mu_\theta^{(i)} : \mathbb{R}^{n_\mathrm{z}+n_\mathrm{u}} \rightarrow \mathbb{R}^{n_\mathrm{y}}$, and $\Sigma_\theta^{{(i)}} : \mathbb{R}^{n_\mathrm{z}+n_\mathrm{u}} \rightarrow \mathbb{K}^{n_\mathrm{y}}$ are all parametrized as feedforward neural networks. The parametrization in \eqref{eq:h_theta} can represent any pdf as $n_\mathrm{G}\to \infty$, as the (weighted) sum of Gaussian distributions is a universal approximator \cite{bishop_mixture_1994,goodfellow_deep_2016}. For computational efficiency, typically three separate ANNs are utilized with $n_\mathrm{G}$ number of outputs to describe the weights $w_\theta$, means $\mu_\theta$, and covariance terms $\Sigma_\theta$, respectively. Moreover, to ensure that $\hat{p}_k^\mathbf{y}$ represents a valid \emph{Gaussian Mixture Model} (GMM), i.e., $\sum_{i=i}^{n_\mathrm{G}} w_\theta^{(i)}=1$, and $w_\theta^{(i)}>0$, $\Sigma_\theta^{(i)} \succ 0$, these restrictions are enforced by applying appropriate activation function in the last layer of each ANN in $h_\theta^\mathrm{z}$. The selection of the model order $n_\mathrm{z}$ is also crucial as it should uniquely characterize $\mathcal{S}_\mathbf{x}$ to acquire an equivalent representation to \eqref{eq:data_gen_sys}. However, as~\cite[Thm. 5.2]{beintema_datadriven_2024} implies, there will always exist an MSS model that converges to the true system as $n_\mathrm{z}$ increases. As a result, there exists a finite MSS model order for all stochastic systems in the form of~\eqref{eq:data_gen_sys} such that a desired level of accuracy is achieved. For further details about the parametrization, data-driven identification approach, and detailed ablation studies, refer to~\cite{beintema_meta-statespace_2024}.

For simulating a given meta-state-space model, the initial condition $z_0$ is required. This is equal to knowing the initial state distribution $p^\mathbf{x}_0$ when forward propagating the Chapman-Kolmogorov equations \eqref{eq:chapman_kolg}. However, in general, the initial state distribution can not be accessed, and typically a state distribution estimator is applied in the form of $p(\mathbf{x}_k\mid u_{k-n:k-1}, y_{k-n:k-1})$, where $n\geq 1$ and $y_{k-n:k-1}$ being the observed samples of ${\mathbf y}_{k-n:k-1} $. In the meta-state-space form, this is interpreted as a meta-state observer, i.e.,
\begin{equation}\label{eq:encoder_fun}
    z_k = \psi(u_{k-n:k-1}, y_{k-n:k-1}),
\end{equation}
where  $\psi$ is the so-called encoder function. For fading memory systems, given that $n$ is sufficiently large, \eqref{eq:encoder_fun} is a consistent reconstructability map for the MSS model. In practice, the encoder is constructed as a parametric model and is co-estimated with the model parameters.

A key advantage of the meta-state-space representation is that $f^\mathrm{z}$ is \emph{deterministic}. Intuitively, the meta-state $z_k$ represents the probability distribution of $\mathbf{x}_k$, such that the evolution of the full stochastic state distribution is described by a deterministic process (given an initial condition). This yields a direct and compact representation of stochastic systems that enables the extension of control methods originally developed for deterministic systems to the stochastic setting. This paper proposes the first model-predictive control design method based on meta-state-space models, yielding an efficient and computationally tractable SMPC approach with recursive feasibility and stability guarantees.

\section{Set-point stabilization}\label{sec:MSS-MPC}

\subsection{Problem formulation}
We consider the control of the stochastic dynamical system represented by \eqref{eq:data_gen_sys}. We impose \emph{input constraints} that are provided as a polytope, i.e., $u_k\in U\subset \mathbb{R}^{n_\mathrm{u}}$, where $U = \{u: H_\mathrm{u} u - h_\mathrm{u} \leq 0\}$, with $H_\mathrm{u}\in\mathbb{R}^{l_\mathrm{u}\times n_\mathrm{u}}$, and $h_\mathrm{u}\in\mathbb{R}^{l_\mathrm{u}}$. Furthermore, we consider box-constraints for the output, in terms of the following \emph{chance constraints}:
\begin{equation}\label{eqs:chance_constraints}
    \mathbb{P}\{\mathbf{y}_k\geq y_\mathrm{min}\} \geq p_\mathrm{min}, \quad \mathbb{P}\{\mathbf{y}_k\leq y_\mathrm{max}\} \geq p_\mathrm{max},
\end{equation}
where $p_\mathrm{min}, p_\mathrm{max}\in[0,1]$ are the desired probability levels of satisfying the lower and upper bounds $y_\mathrm{min}, y_\mathrm{max} \in \mathbb{R}^{n_\mathrm{y}}$ element-wise. We consider a user-specified reference output pdf provided as $p_\mathrm{ref}^\mathbf{y}$. The control objective is to stabilize system \eqref{eq:data_gen_sys} around an equilibrium point which corresponds to an output pdf (approximately) matching the reference pdf. The function $p_\mathrm{ref}^\mathbf{y}$ can also change with time (reference distribution tracking); however, in this section, we only consider the constant reference case (set-point distribution case), and later in Section~\ref{sec:changing_ref}, we will handle time-varying $p_\mathrm{ref}^\mathbf{y}$.

\subsection{General MPC setting}
In the following, we formulate the MPC problem and conduct all theoretical analysis based on the true MSS representation of system \eqref{eq:data_gen_sys}. Consequently, all guarantees are w.r.t. the MSS model\footnote{For deterministic systems, the identified MSS model directly characterizes model uncertainty, unlike conventional deterministic identification methods. Its use in the proposed MPC scheme is similar to the GP-MPC setting, where the quality of uncertainty characterization and the validity of closed-loop guarantees likewise depend on the underlying assumptions, kernel choice, and number of data points.}, and characterization of possible model mismatch due to data-driven identification, together with possible robust/adaptive MPC extensions to directly handle them, are open questions that are outside the scope of the current paper.

For the MPC, we define the finite horizon cost as
\begin{equation}\label{eq:setpoint_cost}
    J_N(z_{0:N}, u_{0:N-1}) \triangleq V_\mathrm{f}(z_N) + \textstyle\sum_{i=0}^{N-1} \ell(z_i, u_i),
\end{equation}
where $\ell:\mathbb{R}^{n_\mathrm{z}}\times \mathbb{R}^{n_\mathrm{u}}\rightarrow\mathbb{R}$ is a user-specified stage cost, $V_\mathrm{f}:\mathbb{R}^{n_\mathrm{z}}\rightarrow\mathbb{R}$ is the terminal cost, and $N\in\mathbb{Z}_+$ is the horizon length. We propose to choose the control input $u_k$ by solving the following \emph{optimal control problem}~(OCP) at each time step $k$:
\begin{subequations}\label{eq:optim_problem_meta}
\begin{align}
\min_{\substack{z_{k:k+N\vert k}\\u_{k:k+N-1\vert k}}} \quad &J_N(z_{k:k+N\vert k}, u_{k:k+N-1\vert k})\\
\textrm{s.t.} \quad &z_{k\vert k} = z_k,\label{eq:mpc_encoder}\\
  & z_{k+i+1\vert k} = f^\mathrm{z}(z_{k+i\vert k}, u_{k+i\vert k}),\\
    &\mathbb{P}\{\mathbf{y}_{k+i\vert k} \geq y_\mathrm{min}\} \geq p_\mathrm{min},\label{eq:chance_const_min}\\
  &\mathbb{P}\{\mathbf{y}_{k+i\vert k} \leq y_\mathrm{max}\} \geq p_\mathrm{max},\label{eq:chance_const_max}\\
    & u_{k+i\vert k} \in U,\label{eq:input_constr}\\
  & z_{k+N\vert k} \in \mathcal{Z}_\mathrm{f},
\end{align}
\end{subequations}
where $i\in\mathbb{I}_0^{N-1}$, $\mathcal{Z}_\mathrm{f}$ is the terminal set, defined later, and the pipe notation ($\vert$) is introduced such that $z_{k+i\vert k}$, $u_{k+i\vert k}$, and $\mathbf{y}_{k+i\vert k}$ are the predicted meta-state, input and output (as a random variable), respectively, at time $k+i$ based on time $k$. After solving the optimization problem \eqref{eq:optim_problem_meta}, the predicted optimal meta-state and input sequences are obtained as $z^\star_{k:k+N\vert k}$, $u^\star_{k:k+N-1\vert k}$. Then the first element of the optimal control sequence is applied to system \eqref{eq:data_gen_sys}, at every time step; hence, formally $u_k = \kappa_\mathrm{MPC}(z_k) = u^\star_{k\vert k}$. The initial state in \eqref{eq:mpc_encoder} could, in practice, be provided by the encoder \eqref{eq:encoder_fun} or alternatively, any other meta-state estimator. Our theoretical guarantees, however, hold only in the case of exact meta-state feedback. For the true MSS model, the utilized encoder ensures this property under the assumptions specified in~\cite{beintema_datadriven_2024}. Otherwise, a similar robustness analysis should be conducted as in~\cite{kohler_simple_2019} regarding the meta-state observer, which is outside the scope of this paper.

This concludes the high-level formulation of the proposed model predictive control approach; however, specifying the required ingredients is challenging. For instance, the stage cost should be selected such that it specifies the difference between the realized output pdf and $p_\mathrm{ref}^\mathbf{y}$, terminal ingredients have to be chosen such that recursive feasibility and stability can be guaranteed, and chance constraints should be evaluated in a computationally efficient manner. These aspects are addressed in the following subsection, constituting Contribution~C1.

\subsection{Stage cost and set-point selection}\label{sec:stage_cost_and_setpoint}
Given a reference output probability density function~$p_\mathrm{ref}^\mathbf{y}$, a straightforward approach is to choose the KL divergence between $p_\mathrm{ref}^\mathbf{y}$ and $h^\mathrm{z}$ as the stage cost:
\begin{equation}\label{eq:stage_cost_KL}
    \ell_{\mathrm{KL}}(z, u) = \mathrm{KL}(p_\mathrm{ref}^\mathbf{y}, h^\mathrm{z}(z, u)),
\end{equation}
where $\mathrm{KL}$ is the Kullback-Leibler divergence. However, implementing \eqref{eq:stage_cost_KL} presents multiple challenges:
\begin{itemize}
    \item Even if $p_\mathrm{ref}^\mathbf{y}$ (and also $h^\mathrm{z}$) is a GMM, the KL divergence has no closed-form solution. Available approximations are typically sampling-based (e.g.,~\cite{cui_comparison_2015}) and thus computationally expensive due to repeated evaluations in \eqref{eq:optim_problem_meta}.
    \item Multiple $(z,u)$ pairs may minimize \eqref{eq:stage_cost_KL}, leading to non-unique solutions, which can cause solver divergence and oscillations.
\end{itemize}

Instead of using the KL divergence directly as the stage cost in the MPC optimization, we perform an offline search to find a reference set-point $(\bar{z}, \bar{u})$ satisfying the following requirements: (\textit{i}) satisfy the chance and input constraints in \eqref{eq:optim_problem_meta}, (\textit{ii}) generate an output probability density function close to the desired $p_\mathrm{ref}^\mathbf{y}$, furthermore, to simplify closed-loop stability analysis (\textit{iii}) be an equilibrium point of the MSS system. Finally, we enforce that (\textit{iv}) the found~$\bar{z}$ is controllable. As there is no analytic solution to this problem, we propose an optimization-based approach, first, to satisfy requirements (\textit{i})--(\textit{iii}):
\begin{subequations}\label{eqs:optim_set_distr_search}
\begin{align}
    \min_{z, u} \quad&\mathrm{KL}(p_\mathrm{ref}^\mathbf{y}, h^\mathrm{z}(z, u)),\\
    \text{s.t.}\quad &f^\mathrm{z}(z, u) = z,\\
    &\text{constraints \eqref{eq:chance_const_min}--\eqref{eq:input_constr} hold,}
\end{align}
\end{subequations}
where the Kullback-Leibler divergence needs to be approximated, as in general, there is no closed-form expression for it. Therefore, we take the Monte-Carlo sampling-based method. First, we limit $p_\mathrm{ref}^\mathbf{y}$ to be a sum of $n_\mathrm{G}^\mathrm{ref}$ number of Gaussians\footnote{Note that this is not a restrictive condition, as the sum of Gaussians is a universal approximator and can express an arbitrary pdf with $n_\mathrm{G}^\mathrm{ref}\to\infty$.}, i.e., a GMM, as $p_\mathrm{ref}^\mathbf{y}=\sum\pi_i \mathcal{N}(\mu^\mathrm{ref}_i, \Sigma^\mathrm{ref}_i)$, where $\pi_i$, $\mu_i^\mathrm{ref}, \Sigma_i^\mathrm{ref}$ ($i\in\mathbb{I}_1^{n_\mathrm{G}^\mathrm{ref}}$) denote the weights, means, and covariance matrices corresponding to each Gaussian component. Before setting up the optimization problem, we draw $M$ number of i.i.d. samples from $p_\mathrm{ref}^\mathbf{y}$, acquiring$\{\tilde{y}_i\}_{i=1}^M$, giving
\begin{equation}\label{eq:KL_approx}
    \mathrm{KL}(p_\mathrm{ref}^\mathbf{y}, h^\mathrm{z}(z, u)) \approx \frac{1}{M}\sum_{i=1}^M \log p_\mathrm{ref}^\mathbf{y}(\tilde{y}_i) / p^\mathbf{y}(\tilde{y}_i\vert z, u).
\end{equation}

Following the above-defined sampling method, with $M\to\infty$ the approximation converges towards the true KL divergence, a property that only holds for the Monte-Carlo sampling-based approximation \cite{hershey_approximating_2007}.

\begin{remark}[Set-point pdf]
    Since the $\mathrm{KL}$ divergence is not a metric, it might be hard to interpret how close the found equilibrium is to the reference pdf after solving \eqref{eqs:optim_set_distr_search}. Denote the optima as $\mathrm{KL^\star}\triangleq \mathrm{KL}(p_\mathrm{ref}^\mathbf{y}, h^\mathrm{z}(\bar{z}, \bar{u}))$, then the total variation distance between $p_\mathrm{ref}^\mathbf{y}$ and $h^\mathrm{z}(\bar{z}, \bar{u})$ is bounded by Pinsker's inequality:
    \begin{equation}\label{eq:Pinskers_ineq}
        \sup_{y \in\mathbb{R}^{n_\mathrm{y}}} \vert p_\mathrm{ref}^\mathbf{y}(y) - h^\mathrm{z}(\bar{z}, \bar{u})(y) \vert \leq \sqrt{\mathrm{KL}^\star / 2}.
    \end{equation}
    Consequently, convergence to the set-point $(\bar{z},\bar{u})$ implies convergence of the output pdf to a ball of radius $\sqrt{\mathrm{KL}^\star/2}$ around $p^\mathbf{y}_\mathrm{ref}$ (in the sense of total variation distance).
\end{remark}

Solving \eqref{eqs:optim_set_distr_search} satisfies (\textit{i})--(\textit{iii}); however, requirement (\textit{iv}) still needs to be handled. Intuitively, even if the realized output pdf $h^\mathrm{z}(\bar{z}, \bar{u})$ is proper, the found set-point might realize a state distribution $\bar{p}^\mathbf{x}$ that is unreachable, e.g., a Dirac-like distribution. Therefore, local controllability should be checked after solving optimization \eqref{eqs:optim_set_distr_search}. Suppose $f^\mathrm{z}$ is twice continuously differentiable, then the following Jacobians can be computed:
\begin{equation}\label{eq:linearized_dyn}
    A^\mathrm{z} = \partial f^\mathrm{z}(\bar{z}, \bar{u}) / \partial z_k, \quad B^\mathrm{z} = \partial f^\mathrm{z}(\bar{z}, \bar{u}) / \partial u_k.
\end{equation}
Then, if linearization around the equilibrium point satisfies the controllability rank condition, i.e.,
\begin{equation}\label{eq:controllability_cond}
    \mathrm{rank}\left(\begin{bmatrix}
        B^\mathrm{z} & A^\mathrm{z}B^\mathrm{z} & \cdots & (A^\mathrm{z})^{n_\mathrm{z}-1}B^\mathrm{z}
    \end{bmatrix}\right) = n_\mathrm{z},
\end{equation}
then, there exists a neighbourhood of $\bar{z}$ for which the nonlinear MSS system is controllable. If the found set-point fails the controllability condition, then optimization \eqref{eqs:optim_set_distr_search} needs to be repeated by excluding the previous $(\bar{z}, \bar{u})$ from the search region.

\begin{remark}[Guaranteed controllability]
    Note that the (local) controllability condition \eqref{eq:controllability_cond} can be enforced by directly embedding it into the optimization problem \eqref{eqs:optim_set_distr_search}, e.g., via Gramian-based metrics or determinant/log-det constraints on the linearized dynamics \cite{pasqualetti_controllability_2014,vandenberghe_determinant_1998}. Such formulations are, however, known to significantly increase computational complexity.
\end{remark}

After acquiring an appropriate set-point, the following stage cost is applied in the formulation of \eqref{eq:optim_problem_meta}:
\begin{equation}\label{eq:stage_cost}
    \ell(z, u) = \|z - \bar{z}\|_Q^2 + \|u - \bar{u}\|_R^2,
\end{equation}
where $Q\in\mathbb{R}^{{n_\mathrm{z}\times n_\mathrm{z}}}$ and $R\in\mathbb{R}^{{n_\mathrm{u}\times n_\mathrm{u}}}$ are both symmetric, positive definite weight matrices. The quadratic stage cost simplifies both the online optimization problem and the upcoming stability proof.

\subsection{The mean-matching case}\label{sec:equilibrium}
Rather than specifying a full reference probability density function, it is often sufficient in practice to prescribe reference values for selected moments of the output distribution. The simplest and most common case is when a reference is prescribed only for the expected output value, such that the control objective is to enforce $\mathbb{E}\{\mathbf{y}_k\}=y_\mathrm{ref}$, with optionally minimizing the output variance, see, e.g., \cite{costa_optimal_2012,hakobyan_risk-aware_2019}. In this case, an appropriate set-point can be found by modifying the cost function in optimization problem \eqref{eqs:optim_set_distr_search}, as follows
\begin{subequations}\label{eqs:optim_set_point_search}
\begin{align}
    &\|y_\mathrm{ref} - \mathbb{E}\{\mathbf{y}\}\|_2^2 + \beta \| \mathrm{Var}\{\mathbf{y}\}\|_2^2,\label{eq:cost_expectation_ref}\\
    &\mathbb{E}\{\mathbf{y}\} = \textstyle\sum_{i=1}^{n_\mathrm{G}} w^{(i)}\mu^{(i)},\\
    &\mathrm{Var}\{\mathbf{y}\} = \textstyle\sum_{i=1}^{n_\mathrm{G}} w^{(i)}\big(\Sigma^{(i)} + (\mu^{(i)} - \mathbb{E}\{\mathbf{y}\})(\cdot)^\top\big),\label{eq:mixture_variance}
\end{align}
\end{subequations}
where $w^{(i)}$, $\mu^{(i)}$, $\Sigma^{(i)}$ are associated with each component in the GMM, according to \eqref{eq:h_theta} (with the dependency on $z,u$ omitted for clarity) and $(\cdot)^\top$ is shorthand for repeating the preceding factor, so $(a)(\cdot)^\top = aa^\top$. Moreover, $\beta\in\mathbb{R}_{\geq 0}$ is a user-specified weight that defines the relative importance of finding a set-point with minimal variance. Setting $\beta=0$ results in the optimization only considering the mean-matching task.

\subsection{Chance constraints}\label{sec:chance_constraints}
Next, we provide a tractable computation method for the chance constraints. Since the output pdf is represented by a GMM, the inequality constraints can be expressed by the following integral form:
\begin{multline}\label{eq:ymax-prob-int}
    \mathbb{P}\{\mathbf{y}_k\leq y_\mathrm{max}\}=\int_{-\infty}^{y_\mathrm{max}} \sum_{i=1}^{n_\mathrm{G}} w^{(i)} \frac{1}{\sqrt{(2\pi)^{n_\mathrm{y}} \det\Sigma^{(i)}}}\\\exp\left(-\frac{1}{2}\left(\upsilon - \mu^{(i)}\right)^\top \left(\Sigma^{(i)}\right)^{-1}\left(\upsilon - \mu^{(i)}\right)\right)~\mathrm{d}\upsilon,
\end{multline}
where $\Sigma^{(i)}$, $w^{(i)}$, and $\mu^{(i)}$ are defined as in \eqref{eq:h_theta}, and dependency on $z$, $u$ is omitted for clarity. However, analytic evaluation of \eqref{eq:ymax-prob-int} is not possible, since there is no closed-form expression for the \emph{cumulative distribution function} (CDF) of a GMM. Several numerical methods can be applied for computing \eqref{eq:ymax-prob-int}, see, e.g., \cite{gen_numerical_1992}. Here, for the sake of computational simplicity, we assume that $\Sigma^{(i)}$ is diagonal for all $i\in\mathbb{I}_1^{n_\mathrm{G}}$. This condition can be directly enforced when estimating \eqref{eqs:meta-ANN-model} by only parametrizing the diagonal terms of each $\Sigma^{(i)}$ matrix. Note that even with all $\Sigma^{(i)}$ being diagonal, the mixture covariance is generally full; see~\eqref{eq:mixture_variance}. Furthermore, a GMM with diagonal covariances is still a universal approximator~\cite{bishop_mixture_1994}. In case of diagonal $\Sigma^{(i)}$, the joint CDF has the form:
\begin{equation}\label{eq:chance_constr_eval_simpl}
    \mathbb{P}\{\mathbf{y}_k \leq y_\mathrm{max}\} = \sum_{i=1}^{n_\mathrm{G}} w^{(i)} \prod_{j=1}^{n_\mathrm{y}}\Phi\left(\frac{y_{\mathrm{max},j}-\mu_j^{(i)}}{\sigma_j^{(i)}}\right),
\end{equation}
where $\Phi$ is the CDF of the standard Gaussian distribution, $y_{\mathrm{max},j}$ is the $j^\mathrm{th}$ element of the vector $y_\mathrm{max}$, and similarly, $\mu_j^{(i)}$ is the $j^\mathrm{th}$ element of $\mu^{(i)}$. The same holds for the lower bound constraint; however, we need to utilize the following relation for scalar outputs: $\mathbb{P}\{\mathbf{y}_k \geq y_\mathrm{min}\} = 1 - \mathbb{P}\{\mathbf{y}_k \leq y_\mathrm{min}\}$. Using that $1-\Phi(\zeta) =\Phi(-\zeta)$, the multivariate case has a similar form as~\eqref{eq:chance_constr_eval_simpl}:
\begin{equation}\label{eq:P_min_max_rels}
    \mathbb{P}\{\mathbf{y}_k \geq y_\mathrm{min}\} = \sum_{i=1}^{n_\mathrm{G}}w^{(i)} \prod_{j=1}^{n_\mathrm{y}}\Phi\left(\frac{\mu_j^{(i)} - y_{\mathrm{min},j}}{\sigma_j^{(i)}}\right),
\end{equation}
where $y_{\mathrm{min},j}$ is defined as $y_{\mathrm{max},j}$. Regarding implementation, $\Phi$ is not available in most automatic differentiation libraries, such as CasADi \cite{andersson_casadi_2019}. We can utilize the more fundamental \emph{Gaussian error function} (erf), which is typically included in these toolboxes, to compute the CDF, as $\Phi(\zeta) = \frac{1}{2} (1 + \mathrm{erf}(\zeta/\sqrt{2}))$.

\begin{remark}[Chance constraint scaling]
    The number of Gaussian components $n_\mathrm{G}$ and output dimension $n_\mathrm{y}$ affect the computational cost of \eqref{eq:optim_problem_meta} only through chance-constraint evaluation. In particular, \eqref{eq:chance_constr_eval_simpl} and \eqref{eq:P_min_max_rels} scale linearly with both dimensions. Thus, chance constraints can be evaluated at a moderate cost even for large mixtures and multi-output systems.
\end{remark}

\subsection{Conditions for stability}\label{sec:stability_conds}
Having defined the stage cost and set-point and addressed the chance constraint evaluation, we now specify conditions to guarantee recursive feasibility and closed-loop stability of the MPC. For simplicity, throughout this subsection we assume that $(\bar{z}, \bar{u})=(0,0)$, satisfying the constraints in \eqref{eqs:optim_set_distr_search}. Considering the set-point as the origin is only a technical condition, since it is sufficient to find any appropriate set-point as discussed in Section~\ref{sec:equilibrium}, then shift the origin of the system to this equilibrium point. In conventional deterministic MPC schemes, recursive feasibility and stability guarantees are typically achieved by an appropriate design of terminal ingredients such that the following holds:

\begin{condition}[Terminal ingredients]\label{assum:terminal_region}
    The terminal region $\mathcal{Z}_\mathrm{f}\subset \mathbb{R}^{n_\mathrm{z}}$ is a compact set containing $\bar{z}$. There exists a local stabilizing controller $\kappa_\mathrm{f}:\mathbb{R}^{n_\mathrm{z}}\rightarrow\mathbb{R}^{n_\mathrm{u}}$ such that $\forall z\in\mathcal{Z}_\mathrm{f}$, the following constraints are satisfied:
    \begin{subequations}
    \begin{align}
        f^\mathrm{z}(z, \kappa_\mathrm{f}(z)) \label{eq:terminal_invariance}&\in\mathcal{Z}_\mathrm{f},\\
        \kappa_\mathrm{f}(z) &\in U,\label{eq:terminal_U_constr}\\
        \mathbb{P}\{\mathbf{y}_k\leq y_\mathrm{max}\vert z, \kappa_\mathrm{f}(z)\} &\geq p_\mathrm{max},\\
        \mathbb{P}\{\mathbf{y}_k\geq y_\mathrm{min}\vert z, \kappa_\mathrm{f}(z)\} &\geq p_\mathrm{min}.\label{eq:terminal_P_constr}
    \end{align}
    Moreover, the terminal cost $V_\mathrm{f}: \mathbb{R}^{n_\mathrm{z}}\rightarrow \mathbb{R}$ is continuous and decreases under the terminal controller $\kappa_\mathrm{f}$ inside $\mathcal{Z}_\mathrm{f}$:
    \begin{equation}\label{eq:terminal_cost_decrease}
        V_\mathrm{f}(f^\mathrm{z}(z, \kappa_\mathrm{f}(z))) - V_\mathrm{f}(z) \leq -\ell(z, \kappa_\mathrm{f}(z)), \, \forall z\in\mathcal{Z}_\mathrm{f}.
    \end{equation}
    \end{subequations}
\end{condition}

Under Condition~\ref{assum:terminal_region}, recursive feasibility and stability of the MPC can be proven by following classical deterministic MPC literature, e.g., \cite{rawlings_model_2017}. However, constructing terminal ingredients that meet these criteria is challenging. Hence, we treat Condition~\ref{assum:terminal_region} as a general requirement for now and provide a suitable design process in Section~\ref{sec:terminal_cond}. To simplify the upcoming theoretical analysis, we assume that the following mild condition holds

\begin{condition}[Continuity and compactness]\label{assum:f_cont}
    The meta-state transition function is twice continuously differentiable, i.e., $f^\mathrm{z}\in\mathcal{C}_2$. The input constraint set $U$ is compact. The set satisfying the chance constraints, i.e., $\mathbb{T}=\{(z,u) \vert \text{constraints \eqref{eq:chance_const_min}--\eqref{eq:chance_const_max}}\}$, is closed. 
\end{condition}

Keep in mind that the continuous differentiability conditions in Condition~\ref{assum:f_cont} are not restrictive, as $f^\mathrm{z}$ is parametrized as an ANN, and commonly applied $\mathcal{C}_2$ activation functions (such as hyperbolic tangent, logistic, etc.) ensure the $\mathcal{C}_2$ property of the network.

Next, we define the optimal cost under the optimal input sequence generated by solving \eqref{eq:optim_problem_meta} at time step $k$, as
\begin{equation}\label{eq:optim_cost_val}
    J_{N}^\star(z_k) \triangleq J_N(z_{k:k+N\vert k}^\star, u_{k:k+N-1\vert k}^\star),
\end{equation}    
then, formulate the following controllability condition:

\begin{condition}[Weak controllability]\label{assum:controllability}
    There exists a function $\alpha_2\in\mathcal{K}_\infty$ such that for any meta-state for which \eqref{eq:optim_problem_meta} is feasible, i.e., $\forall z\in\mathcal{Z}_0$, it holds: $J_{N}^\star\leq \alpha_2(\|z\|_2)$.
\end{condition}

Condition~\ref{assum:controllability} is common for MPC stability proofs to guarantee that $J_{N}^\star$ is a Lyapunov function, see, e.g., \cite{kohler_analysis_2021}. Furthermore, it is only a technical condition as $\bar{z}\in\mathrm{int}(\mathcal{Z}_\mathrm{f})$ trivially satisfies Condition~\ref{assum:controllability} \cite{grune_nonlinear_2017}, where $\mathrm{int}(\mathcal{Z}_\mathrm{f})$ denotes the interior of set $\mathcal{Z}_\mathrm{f}$. Lastly, the upcoming stability proof requires the following condition regarding the stage cost:

\begin{condition}[Positive definite stage cost]\label{property:pos_def_stage_cost}
    The stage cost $\ell$ is a continuous function and it satisfies $\ell(\bar{z},\bar{u})=0$, while $\ell(z, u)>0$ for all $(z, u)\in\mathbb{R}^{n_\mathrm{z}}\times\mathbb{R}^{n_\mathrm{u}}\neq (\bar{z},\bar{u})$. Moreover, there exists an $\alpha_1\in\mathcal{K}_\infty$ function, such that $\ell(z, u)\geq \alpha_1(\|z-\bar{z}\|_2)$.
\end{condition}

Note that the applied stage cost \eqref{eq:stage_cost} trivially satisfies Condition~\ref{property:pos_def_stage_cost} under $Q,R\succ0$. Now that we have specified all the necessary conditions, recursive feasibility and asymptotic stability of the closed-loop system under the MPC scheme \eqref{eq:optim_problem_meta} can be proven straightforwardly, following conventional results from the literature.

\begin{theorem}[Recursive feasibility and stability]\label{thm:rec_feas_and_stab}
    Assume that Conditions~\ref{assum:terminal_region}--\ref{property:pos_def_stage_cost} are satisfied. Suppose the MPC problem \eqref{eq:optim_problem_meta} is feasible at time instance $k=0$ with $z_0\in\mathcal{Z}_0$, then the OCP \eqref{eq:optim_problem_meta} is feasible $\forall k> 0$ and the origin is an asymptotically stable equilibrium point for the closed loop system $z_{k+1}=f^\mathrm{z}(z_k, \kappa_\mathrm{MPC}(z_k))$.
\end{theorem}
\begin{proof}
    The proof of~\cite[Thm.~2.19]{rawlings_model_2017} directly implies Theorem~\ref{thm:rec_feas_and_stab}. Note that the set-point selection of Sec.~\ref{sec:equilibrium} combined with Condition~\ref{assum:terminal_region} fulfills Assump.~2.14 in \cite{rawlings_model_2017}.
\end{proof}

Consequently, Theorem~\ref{thm:rec_feas_and_stab} proves that the proposed MSS-based SMPC scheme results in a stable closed-loop system. As discussed, Conditions \ref{assum:f_cont} and \ref{assum:controllability} can be viewed as technical conditions, while Condition~\ref{property:pos_def_stage_cost} holds by construction. However, Condition~\ref{assum:terminal_region} can only be guaranteed by careful selection of the terminal ingredients; thus, in the following, we present a design process under which it is guaranteed to hold, corresponding to Contribution~C2.

\subsection{Terminal ingredients}\label{sec:terminal_cond}
As discussed in Section~\ref{sec:stability_conds}, guaranteeing recursive feasibility and closed-loop stability of the proposed SMPC scheme relies on the terminal ingredients satisfying Condition~\ref{assum:terminal_region}. Classical deterministic MPC formulations use invariant terminal sets and quadratic terminal costs based on locally stabilizing linear feedback controllers, typically obtained from LQR design, which provide well-understood stability guarantees~\cite{chen_quasi-infinite_1998,rawlings_model_2017}. Beyond these methods, several extensions have been developed, such as enlarged terminal regions~\cite{limon_enlarging_2005}, generalized constraints~\cite{fagiano_generalized_2013}, learning-based constructions~\cite{abdufattokhov_learning_2024}, etc. While these approaches can reduce conservatism or improve performance, an LQR-based terminal design remains appealing due to its simplicity and theoretical guarantees. Accordingly, we adopt a standard LQR-based design process to ensure recursive feasibility and stability.

Since the linearized system representation \eqref{eq:linearized_dyn} is stabilizable around $\bar{z}$ according to criteria (\textit{iv}) for the set-point and the rank-condition \eqref{eq:controllability_cond}, a local linear quadratic control policy can be synthesized around $(\bar{z}, \bar{u})$. Let us denote the optimal LQ feedback gain with $K$, and the solution of the \emph{discrete-time algebraic Ricatti equation} (DARE) associated with the LQR as $P$. Hence,
\begin{equation}\label{eq:LQ_Lyapunov_eq}
    P = (A^\mathrm{z} + B^\mathrm{z}K)^\top P (A^\mathrm{z} + B^\mathrm{z}K) + Q_\varepsilon + K^\top R K,
\end{equation}
where $R$ is the same weight matrix as in \eqref{eq:stage_cost} and $Q_\varepsilon=Q+\varepsilon I_{n_\mathrm{z}}$ with $\varepsilon$ being a small constant introduced to account for the linearization error of $(A^\mathrm{z}, B^\mathrm{z})$ compared to the nonlinear dynamics in $f^\mathrm{z}$. After solving \eqref{eq:LQ_Lyapunov_eq}, we apply a quadratic terminal cost as $V_\mathrm{f}(z) = (z - \bar{z})^\top P (z - \bar{z})$ and a terminal set as $\mathcal{Z}_\mathrm{f} = \{z: V_\mathrm{f}(z) \leq \gamma\}$.

Next, we show that there always exists a sufficiently small $\gamma$ such that the resulting terminal ingredients satisfy Condition~\ref{assum:terminal_region}. First, note that the chance constraints can be reformulated as general inequality constraints as $g_\mathrm{min}(z,u)\leq 0$ and $g_\mathrm{max}(z,u)\leq 0$, where
\begin{subequations}
\begin{align}
    g_\mathrm{min}(z, u) &\triangleq p_\mathrm{min} -\mathbb{P}\{\mathbf{y}_k\geq y_\mathrm{min}\mid z, u\},\label{eq:g_min}\\
    g_\mathrm{max}(z, u) &\triangleq p_\mathrm{max} -\mathbb{P}\{\mathbf{y}_k\leq y_\mathrm{max}\mid z, u\}.\label{eq:g_max}
\end{align}
\end{subequations}

Finally, the following lemma proves that the terminal ingredients satisfy Condition~\ref{assum:terminal_region}.

\begin{lemma}[Terminal region]\label{lem:terminal_ingredients}
    Let Assumption~\ref{assum:f_cont} hold and assume the chance constraints \eqref{eq:g_min} and \eqref{eq:g_max} are (locally) Lipschitz continuous around $(\bar{z}, \bar{u})$ with Lipschitz constants of $L_{g_\mathrm{min}}$ and $L_{g_\mathrm{max}}$, respectively. Suppose $(A^\mathrm{z}, B^\mathrm{z})$ in \eqref{eq:linearized_dyn} is stabilizable, and $Q, R\succ 0$. If $P$, $K$ matrices satisfy \eqref{eq:LQ_Lyapunov_eq}, then there exists a sufficiently small $\gamma$, such that the terminal cost $V_\mathrm{f}=(z-\bar{z})^\top P(z-\bar{z})$, terminal controller $\kappa_\mathrm{f}(z)=\bar{u}+K(z-\bar{z})$, and terminal set $\mathcal{Z}_\mathrm{f}=\{z:V_\mathrm{f}(z)\leq\gamma\}$ satisfy Condition~\ref{assum:terminal_region}.
\end{lemma}
\begin{proof}\renewcommand{\qedsymbol}{}
    See Appendix~\ref{sec:appendix_lemma_terminal_ing_proof}.
\end{proof}

As $w^\mathrm{(i)}$, $\mu^{(i)}$, and $\Sigma^{(i)}$ are parametrized by Lipschitz continuous ANNs, $g_\mathrm{min}$ and $g_\mathrm{max}$ are compositions of Lipschitz functions. Hence, assuming \eqref{eq:g_min} and \eqref{eq:g_max} are Lipschitz is not restrictive. The proof of Lemma~\ref{lem:terminal_ingredients} provides an analytic expression for $\gamma$ to satisfy Condition~\ref{assum:terminal_region}, mainly using (local) Lipschitz continuity arguments. However, computing the local Lipschitz bounds can result in a conservative $\gamma$ level-set, which might cause numerical difficulties when solving \eqref{eq:optim_problem_meta}. As an alternative, Appendix~\ref{appendix:terminal_ingredients} presents a sampling-based, iterative design process to obtain a (potentially) less conservative bound for the terminal region.

\begin{remark}[Terminal equality constraint]
A terminal equality condition, i.e., $\mathcal{Z}_\mathrm{f}=\{\bar{z}\}$, can also be applied without additional design steps. However, that might result in a higher computational demand or infeasibility on a given horizon, due to the optimization problem having significantly fewer degrees of freedom.
\end{remark}

\subsection{Reachability of the terminal region}\label{sec:reachability}
One aspect of the terminal ingredients is left unexplored, namely, the $N$-step reachability of the terminal region. Recursive feasibility and stability of the MSS-based SMPC scheme can only be proven for such initial $z_0$ values from which the terminal region is reachable in $N$ steps. Thus, next, we present a systematic approach for the verification of the reachability of the terminal region. The $N$-step backward reachable set associated with the applied terminal region is formally defined as:

\begin{definition}[$N$-step backward reachable set]
    Given the MPC problem \eqref{eq:optim_problem_meta} with a terminal set $\mathcal{Z}_\mathrm{f}$, the $N$-step backward reachable set $\mathcal{Z}_0^{(N)}$ of $\mathcal{Z}_\mathrm{f}$ is the set of initial points $z_0\in\mathbb{R}^{n_\mathrm{z}}$ for which there exists an admissible input sequence $\{u_k\}_{k=0}^{N-1}$ such that $\forall k\in\mathbb{I}_0^{N-1}$ $z_{k+1}=f^\mathrm{z}(z_k, u_k)$, $\{(z_k, u_k)\}_{k=0}^{N-1}$ satisfy constraints \eqref{eq:chance_const_min}--\eqref{eq:input_constr} and $z_N\in\mathcal{Z}_\mathrm{f}$.
\end{definition}

Equivalently, $\mathcal{Z}_0^{(N)}$ is the feasible set of the finite-horizon MPC problem with terminal constraints $\mathcal{Z}_\mathrm{f}$. As $\mathcal{Z}_0^{(N)}$ is generally non-convex and not computable analytically, we adopt a sampling-based verification approach from \cite{bobiti_automated-sampling-based_2018} to construct a formal inner approximation of $\mathcal{Z}_0^{(N)}$. First, we introduce a notation corresponding to the optimal cost for the propagated meta-state value under the MSS-MPC policy, similarly to \eqref{eq:optim_cost_val}, as
\begin{equation}
    J_{N}^{\star,+}(z_k) \triangleq J_{N}^\star(f^\mathrm{z}(z_k, \kappa_\mathrm{MPC}(z_k))).
\end{equation}
Then, we introduce the property function, as
\begin{equation}\label{eq:property_function_F}
    F(z) = \begin{cases}
        J_{N}^{\star,+}(z) - \rho_\mathrm{c} J_{N}^\star(z),\text{ if \eqref{eq:optim_problem_meta} is feasible,}\\
        F_\mathrm{inf},\text{ otherwise},
    \end{cases}
\end{equation}
where $\rho_\mathrm{c}\in(0,1]$ is a contraction factor, and $F_\mathrm{inf}\in\mathbb{R}_{>0}$ is a constant associated with the infeasible points $z\in\mathbb{R}^{n_\mathrm{z}}\backslash \mathcal{Z}_0^{(N)}$. By standard MPC theory (see Section~\ref{sec:stability_conds}), $J_{N}^\star$ satisfies the decrease condition in \eqref{eq:property_function_F} for all $z\in\mathcal{Z}_0^{(N)}$, i.e., $F(z)\leq 0$ over $\mathcal{Z}_0^{(N)}$. Next, we search for a set $\mathcal{A}\subset\mathbb{R}^{n_\mathrm{z}}$ s.t. for all $z\in\mathcal{A}$, $F(z)\leq 0$, thus $\mathcal{A}$ provides a certified approximation for the $N$-step reachable set. 

Following \cite{bobiti_automated-sampling-based_2018}, the search space, i.e., the \emph{region of interest} (ROI), of meta-state space is covered by a collection of axis-aligned hyper-rectangles $\mathcal{B}_{\delta_{z_\mathrm{s}}}(z_\mathrm{s})$ (see Fig.~\ref{fig:DOA_illustr} for illustration), centered at sample points $z_\mathrm{s}$, where $\delta_{z_\mathrm{s}}\in\mathbb{R}^{2n_\mathrm{z}}$ contains the half-widths of the hyper-rectangles along each axis. For each sample $z_\mathrm{s}$, the method verifies the following condition
\begin{subequations}\label{eqs:DOA_tightened_cond}
\begin{align}
    F(z_\mathrm{s})&\leq -\eta(z_\mathrm{s}, \delta_{z_\mathrm{s}}),\\ \eta(z_\mathrm{s}, \delta) &\triangleq a_{z_\mathrm{s}} \max\{\vert\delta_{z_\mathrm{s}}\vert\} + b_{z_\mathrm{s}} + \epsilon(z_\mathrm{s}),
\end{align}
\end{subequations}
where $a_{z_\mathrm{s}}$ bounds the Lipschitz constant of $F$ within the hyper-rectangle~$\mathcal{B}_{\delta_{z_\mathrm{s}}}(z_\mathrm{s})$, computed by using the first order Taylor approximation of $F$ around $z_\mathrm{s}$, while $b_{z_\mathrm{s}}$ bounds the Lagrange remainder of this approximation. Furthermore, $\epsilon(z_\mathrm{s})$ accounts for the magnitude of any jump discontinuity of $F$ inside $\mathcal{B}_{\delta_{z_\mathrm{s}}}(z_\mathrm{s})$.  For exact calculations of the coefficient $a_{z_\mathrm{s}}$, $b_{z_\mathrm{s}}$, $\epsilon(z_\mathrm{s})$, see~\cite{bobiti_automated-sampling-based_2018}. As a result, satisfying the (sampling-based) condition $F(z_\mathrm{s})\leq-\eta(z_\mathrm{s}, \delta_{z_\mathrm{s}})$ guarantees that $F(z)\leq 0$ holds for all $z$ in the hyper-rectangle $\mathcal{B}_{\delta_{z_\mathrm{s}}}(z_\mathrm{s})$. Rectangles that fail the condition are recursively refined and re-evaluated\footnote{See \cite[Algorithm 1]{bobiti_automated-sampling-based_2018} for a detailed description of the refinement and re-evaluation process.}, down to a user-specified minimum size $\delta_\mathrm{min}$. The union of all rectangles that pass the condition forms the set $\mathcal{A}$. The results of~\cite{bobiti_automated-sampling-based_2018} require $F$ to be piecewise continuous and twice differentiable on each continuous region. Applying $\mathcal{C}_\infty$ activation functions (e.g., hyperbolic tangent) and a parametrization that guarantees positive definite GMM covariance matrices renders all functions in \eqref{eq:optim_problem_meta} smooth. We further assume that, for all $z$ in the ROI where \eqref{eq:optim_problem_meta} is feasible, its minimizer is unique and satisfies the linear independence constraint qualification and the strong second-order sufficient condition. Then, $\kappa_\mathrm{MPC}$ is locally Lipschitz and $J_N^\star$ is continuously differentiable, and both are smooth~\cite{bonnans2013perturbation}. Consequently, $F$ is piecewise continuous and two-times differentiable, and set $\mathcal{A}$ provides a certified inner approximation of the $N$-step backward reachable set, i.e., $\mathcal{A}\subseteq \mathcal{Z}_0^{(N)}$, following Theorem III.4 of~\cite{bobiti_automated-sampling-based_2018}.

\begin{figure}
    \centering
    \includegraphics{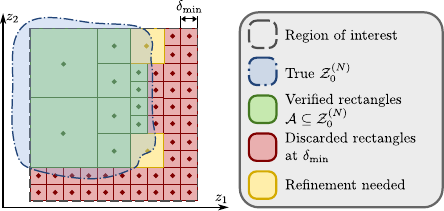}\vspace{-4pt}
    \caption{Illustration of the sampling-based approximation of $\mathcal{Z}_0^{(N)}$. The dots in each rectangle represent the samples $z_\mathrm{s}$.}\vspace{-4pt}
    \label{fig:DOA_illustr}
\end{figure}

If the $N$-step backward reachable set of the associated terminal region is too small, one could synthesize an additional set-point $(\bar{z}_2, \bar{u}_2)$. Repeating the previously discussed process, first an appropriate terminal region $\mathcal{Z}_{\mathrm{f},2}$, then a feasible set $\mathcal{Z}_{0,2}^{(N)}$ can be constructed, resulting in the combined feasible set $\mathcal{Z}_{0}^{(N)}\cup \mathcal{Z}_{0,2}^{(N)}$. This can be repeated until the combined feasible set is large enough, see Fig.~\ref{fig:reachable_sets_and_PDFs}. When the MPC is initialized, an index $j$ is first determined such that $z_0 \in \mathcal{Z}^{(N)}_{0,j}$. The setpoint associated with $\mathcal{Z}^{(N)}_{0,j}$ is then selected as the reference.

\begin{figure}
    \centering
    \includegraphics[width=\linewidth]{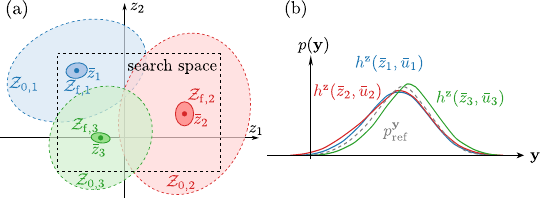}\vspace{-4pt}
    \caption{Schematic representation of the backward reachable sets associated with each set-point in Panel (a). Panel (b) illustrates the possibility of multiple output pdfs (and hence set-points) generating similar distributions as the reference.}\vspace{-4pt}
    \label{fig:reachable_sets_and_PDFs}
\end{figure}

\section{Online changing reference}\label{sec:changing_ref}

\subsection{Reference tracking MPC formulation}\label{sec:ref_tracking_MPC}
In many practical applications, the control objective may change during operation. Consequently, the offline-determined set-point $(\bar{z}, \bar{u})$ may need to be recomputed. However, finding an appropriate reference point and redesigning the terminal ingredients is time-consuming, often requiring manual adjustments by the user and thus, is typically not feasible online. Since several deterministic MPC formulations can efficiently address reference tracking applications, we aim to extend these methodologies to the MSS-based SMPC approach. Most approaches, such as \cite{limon_nonlinear_2018}, introduce an artificial reference $r=\mathrm{vec}(z_\mathrm{r}, u_\mathrm{r})$ to ensure closed-loop stability even if the desired control objective is not achievable. This formulation requires an additional cost term $V_\mathrm{r}$, which, in the current setting, penalizes the discrepancy between the reference $p_{\mathrm{ref}}^\mathbf{y}(k)$ and the output pdf corresponding to the artificial reference, i.e., $h^\mathrm{z}(z_\mathrm{r}, u_\mathrm{r})$. In practice, $V_\mathrm{r}$ can be implemented as a sampling-based approximation of the KL divergence. This contrasts with the formulation proposed in Sect.~\ref{sec:stage_cost_and_setpoint}, where a KL divergence-based stage cost was omitted partly due to its computational burden. Reference tracking, however, forms a fundamentally different control problem compared to set-point stabilization and therefore necessitates a computationally more demanding formulation. This reflects a trade-off between the ability to handle online-changing references and maintaining computational simplicity.

With an appropriate choice of offset cost function, recursive feasibility of the resulting MPC problem can be proven under mild assumptions, see, e.g., \cite{limon_nonlinear_2018,kohler_nonlinear_2020b}. Proving exponential stability, however, would require convexity and uniqueness conditions that cannot be guaranteed in the present setting due to the GMM structure of $h^\mathrm{z}$. In this context, the zone-tracking formulation of \cite{soloperto_nonlinear_2023} is particularly attractive as it enables optimization over a generally non-convex set-point manifold. Motivated by this property, we formulate an MSS-based tracking MPC scheme that guarantees stability without terminal conditions, by extending the results of~\cite{soloperto_nonlinear_2023}.

First, we reformulate the MPC cost elements to formally handle online changing setpoints. Since we aim to avoid terminal ingredients, the cost function now consists solely of stage-cost evaluations:
\begin{align}
    &J_N(z_{k:k+N-1\vert k}, u_{k:k+N-1\vert k}, z_\mathrm{r}, u_\mathrm{r}) =\nonumber\\ &\qquad\qquad\qquad\qquad \textstyle\sum_{i=0}^{N-1} \ell(z_{k+i\vert k}, u_{k+i\vert k}, z_\mathrm{r}, u_\mathrm{r}),\\
    &\ell(z,u) = \|z-z_\mathrm{r}\|_Q^2 + \|u-u_\mathrm{r}\|_R^2,\label{eq:stage_cost_tracking}
\end{align}
with $Q,R\succ 0$ both being symmetric, positive definite weight matrices. Note that the artificial reference should also satisfy requirements (\textit{i})--(\textit{iii}), similar to when choosing $(\bar{z}, \bar{u})$ in the offline case. Requirements (\textit{i}), (\textit{iii}) will be addressed by additional constraints in the optimization problem, while the offset cost will be selected to guarantee satisfaction of  (\textit{ii}). Then, for a given reference $r=\mathrm{vec}(z_\mathrm{r}, u_\mathrm{r})$, the optimal stage cost is defined as
\begin{subequations}
\begin{align}
    \ell_\star(z, z_\mathrm{r}, u_\mathrm{r}) &\triangleq \min_{u}\,\,\ell(z,u,z_\mathrm{r}, u_\mathrm{r})\\
    \text{s.t.}\quad& \text{constraints \eqref{eq:chance_const_min}--\eqref{eq:input_constr} hold}.
\end{align}
\end{subequations}
Let us consider a user-defined offset cost $V_\mathrm{r} : \mathbb{R}^{n_\mathrm{z}} \times \mathbb{R}^{n_\mathrm{u}}\rightarrow \mathbb{R}$ that quantifies the deviation between the reference $p_\mathrm{ref}^\mathbf{y}(k)$ and the realized output pdf associated with the artificial reference, i.e., $h^\mathrm{z}(z_\mathrm{r}, u_\mathrm{r})$. Finally, we denote the set of feasible set-points as
\begin{multline}
    \mathbb{S} = \{\mathrm{vec}(z,u)\in\mathbb{R}^{n_\mathrm{z}+n_\mathrm{u}}\mid z=f^\mathrm{z}(z,u),\\ \text{constraints \eqref{eq:chance_const_min}--\eqref{eq:input_constr} hold}\}.
\end{multline}
Then, we formulate the following zone-tracking MPC with the meta-state-space representation:
\begin{subequations}\label{eq:tracking_ocp}
\begin{align}
\min\nolimits_{z_{\cdot|k},u_{\cdot|k},z_\mathrm{r},u_\mathrm{r}}\ & J_N(z_{\cdot|k},u_{\cdot|k},z_\mathrm{r},u_\mathrm{r}) + V_\mathrm{r}(z_\mathrm{r},u_\mathrm{r})\\
\text{s.t.}\ & \text{constraints \eqref{eq:mpc_encoder}--\eqref{eq:input_constr} hold},\\
& \mathrm{vec}(z_\mathrm{r},u_\mathrm{r})\in\mathbb{S},\label{eq:setpoint_constr}\\
& J_N(z_{\cdot|k},u_{\cdot|k},z_\mathrm{r},u_\mathrm{r}) < \bar{J}_N,\label{eq:stability_constr}
\end{align}
\end{subequations}
with $i\in\mathbb{I}_0^{N-1}$, and the notations $z_{\cdot|k} \triangleq z_{k:k+N-1|k}$, $u_{\cdot|k} \triangleq u_{k:k+N-1|k}$ are introduced for simplicity. Solving \eqref{eq:tracking_ocp} not only provides an optimal meta-state and input sequence, but also the optimal reference point, $\mathrm{vec}(z_\mathrm{r}^\star, u_\mathrm{r}^\star)$ for all $k\geq 0$, denoted as $r_k^\star$. The constraint \eqref{eq:stability_constr}, similar to the terminal constraints in Sect.~\ref{sec:MSS-MPC}, ensures stability of the tracking MPC formulation. This bound will be discussed in more detail later, in Sect.~\ref{sec:changing_ref_stab}. Moreover, similar to Sect.~\ref{sec:MSS-MPC}, the MPC problem \eqref{eq:tracking_ocp} uses the true MSS representation of system \eqref{eq:data_gen_sys} and the exact meta-state values as initial conditions; thus, all subsequently derived guarantees are model-based. A robust extension of the approach is outside the scope of this paper, but is an important future direction.

\subsection{Offset cost selection}
Having defined the zone-tracking MPC formulation, we now formulate requirements for the offset cost. First, let us denote the minimum as $V^\mathrm{min}_\mathrm{r} \triangleq \min_{r\in\mathbb{S}} V_\mathrm{r}(r)$,
and the set of optimal setpoints that achieve the minima as $\mathbb{S}_\mathrm{o} = \{r\in\mathbb{S}\mid V_\mathrm{r}(r) = V_\mathrm{r}^\mathrm{min}\}$. Then, let $C(r, \mathbb{S}_\mathrm{o})$ denote the set of continuous curves that start from any $r\in\mathbb{S}$ and end at any $r_\mathrm{o}\in\mathbb{S}_\mathrm{o}$, and belong to $\mathbb{S}$. Given a set-point $r_1\in\mathbb{S}$, the distance curve $d_{\mathbb{S}_\mathrm{o}} : \mathbb{S} \rightarrow \mathbb{R}_{\geq 0}$ is defined as $d_{\mathbb{S}_\mathrm{o}}(r_1) \triangleq \min_{c\in C(r_1, \mathbb{S}_\mathrm{o})} \mathrm{Length}(c)$, where $\mathrm{Length}$ denotes the length of curve $c$. Based on~\cite{soloperto_nonlinear_2023}, the offset cost should satisfy:

\begin{condition}[Curve distance]\label{cond:curvatures}
     There exists constants $k_0, k_1 > 0$, such that for any $r= (z_\mathrm{r}, u_\mathrm{r})\in\mathbb{S}$ and $\epsilon\in\left[0, 1\right]$, there exists an $\hat{r}\in\mathbb{S}$, satisfying
    \begin{subequations}
    \begin{align}
        \|\hat{r} - r\|_2 &\leq k_0 \epsilon d_{\mathbb{S}_\mathrm{o}}(r),\label{eq:curvature_assum_1}\\
        V_\mathrm{r}(\hat{r}) - V_\mathrm{r}(r) &\leq -k_1 \epsilon d_{\mathbb{S}_\mathrm{o}}^2(r).\label{eq:curvature_assum_2}
    \end{align}
    \end{subequations}
\end{condition}

If Condition~\ref{cond:curvatures} is fulfilled, recursive feasibility and stability of the MPC scheme can be proven under mild assumptions following~\cite{soloperto_nonlinear_2023}. However, verifying that the offset cost satisfies Condition~\ref{cond:curvatures} is non-trivial. For instance, in \cite{soloperto_nonlinear_2023}, it is shown that choosing $V_\mathrm{r}(r) = d_{\mathbb{S}_\mathrm{o}}^2(r)$ satisfies this requirement. However, implementing this cost in practice requires an explicit characterization of $\mathbb{S}_\mathrm{o}$, which is unrealistic in the considered setting (or in general for most real-world applications). Even if such a characterization were available, evaluating the resulting expression would be computationally demanding. Instead, we formulate the offset cost as
\begin{equation}\label{eq:selecetd_offset_cost}
    V_\mathrm{r}(r, k) = \mathrm{KL}(p_\mathrm{ref}^\mathbf{y}(k), h^\mathrm{z}(r)) + \lambda\|r-r_\mathrm{c}\|_2^2,
\end{equation}
where $\lambda\in\mathbb{R}_{>0}$ is a user-specified regularization weight and
$r_\mathrm{c}\in\mathbb{S}$ is a user-chosen feasible set-point. Moreover, $V_\mathrm{r}$ depends on $k$, as $p_\mathrm{ref}^\mathbf{y}$ can change with time, but this dependence is not denoted explicitly in the remainder. In Sect.~\ref{sec:MSS-MPC}, we argued that online evaluation (approximation) of the KL divergence term is computationally demanding. However, now it only needs to be computed once per cost function evaluation compared to the $N$-time evaluation when the KL term is included in the stage cost. The regularization center $r_\mathrm{c}$ can, in principle, be selected freely in $\mathbb{S}$, e.g., at the center of the expected operating region. Alternatively, we recommend choosing the minimum-norm feasible set-point, i.e., $\|r_\mathrm{c}\|_2=\min_{r\in\mathbb{S}}\|r\|_2$. Then, the introduced penalty selects a minimum-norm reference among KL-equivalent ones, thereby conditioning the problem numerically. The following lemma formally proves that the proposed offset cost satisfies Condition~\ref{cond:curvatures} around $\mathbb{S}_\mathrm{o}$.

\begin{lemma}[Satisfaction of distance condition]\label{lem:offset_cost}
    Let $\mathbb{S}\subset \mathbb{R}^{n_\mathrm{z}+n_\mathrm{u}}$ be a connected $\mathcal{C}_2$ manifold and let the offset cost $V_\mathrm{r}:\mathbb{R}^{n_\mathrm{z}+n_\mathrm{u}}\rightarrow\mathbb{R}_{\geq 0}$ chosen as \eqref{eq:selecetd_offset_cost} with $r_\mathrm{c}\in\mathbb{S}$. Suppose that $V_\mathrm{r}$ is $\mathcal{C}_2$ on an open neighborhood of $\mathbb{S}$. Then, there exists a sufficiently large $\lambda>0$ regularization coefficient that ensures satisfaction of Condition~\ref{cond:curvatures} in a compact neighbourhood $\breve{\mathbb{S}}$ of $\mathbb{S}_\mathrm{o}$ such that $\mathbb{S}_\mathrm{o}\subset \breve{\mathbb{S}}\subseteq \mathbb{S}$.
\end{lemma}
\begin{proof}\renewcommand{\qedsymbol}{}
See Appendix~\ref{sec:appendix:offset_cost_proof}.
\end{proof}

The proof of Lemma~\ref{lem:offset_cost} introduces $\phi(\rho)=V_\mathrm{r}(c^\ast(\rho))$, where $c^\ast\in C(r,\mathbb{S}_\mathrm{o})$ denotes the shortest curve on $\mathbb{S}$ from $r$ to $\mathbb{S}_\mathrm{o}$, parameterized by arc-length $\rho$. A key requirement is $\phi''(\rho)>0$, which explains the role of the regularization in the offset cost, as the KL term alone can exhibit flat directions along which $\phi''(\rho)=0$. A possible source of such directions is the permutation symmetry of the GMM implementation of $h^\mathrm{z}$. For instance, any transformation $r_1\mapsto r_2$ that interchanges two Gaussian components, satisfies $h^\mathrm{z}(r_1)=h^\mathrm{z}(r_2)$ and therefore leaves the KL term unchanged. If $r_1$ and $r_2$ are connected by a continuous path in $\mathbb{S}$ along which the GMM components continuously interchange, then $\phi'(\rho)=\phi''(\rho)=0$ along this curve. Note that the required $\lambda$ is small whenever the reference is realizable, i.e., $h^\mathrm{z}(r_\ast)= p_\mathrm{ref}^\mathbf{y}$ for some $r_\ast\in\mathbb{S}$. In such cases, the KL term is nearly convex along $c^\ast$, as its second derivative reduces to the pull-back of the Fisher information metric at $r_\ast$, which is positive semi-definite. Hence, the regularization minimally biases the solution in the practically relevant cases.

The same reasoning can be applied in case the expected output is provided as a reference for the mean, which necessitates an offset cost function of
\begin{equation}
    V_\mathrm{r}(r) = \|y_\mathrm{ref} - \mathbb{E}\{h^\mathrm{z}(r)\}\|_2^2 + \beta\|\mathrm{Var}\{h^\mathrm{z}(r)\}\|_2^2.\
\end{equation}
Flat directions even more severely influence this moment-based cost formulation compared to \eqref{eq:selecetd_offset_cost}, as multiple different pdfs can generate the same mean and variance. Hence, similarly to the KL divergence case, the introduction of the regularization term is required.

\subsection{Stability analysis}\label{sec:changing_ref_stab}
Under Lemma~\ref{lem:offset_cost}, the recursive feasibility and stability of the closed-loop system can be proven based on~\cite{soloperto_nonlinear_2023}. For the sake of completeness, we now adapt the proof in~\cite{soloperto_nonlinear_2023} for the MSS-based tracking SMPC formulation. Before the proof, we define the optimal cost under the optimal input sequence similarly as in \eqref{eq:optim_cost_val}, i.e., $J_{N}^\star(z_k, z_\mathrm{r}, u_\mathrm{r})=J_N(z^\star_{k:k+N-1\vert k}, u^\star_{k:k+N-1\vert k}, z_\mathrm{r}, u_\mathrm{r})$, which, naturally, now depends also on the applied artificial reference. Note that $J_{N}^\star$ does not include the offset cost term $V_\mathrm{r}$, therefore, we introduce the term $H_{N}^\star(z_k)\triangleq J_{N}^\star(z_k, z_\mathrm{r}^\star, u_\mathrm{r}^\star) + V_\mathrm{r}(z_\mathrm{r}^\star, u_\mathrm{r}^\star)$.

\begin{theorem}[Stability of the tracking MPC]\label{thm:tracking_stability}
    Let Condition~\ref{cond:curvatures} hold on $\breve{\mathbb{S}}$. Suppose that there exist constants $\tau, \eta>0$, such that for any $(z_\mathrm{r}, u_\mathrm{r})\in\breve{\mathbb{S}}$, $N\in\mathbb{Z}_+$, and any $z\in\mathbb{R}^{n_\mathrm{z}}$ satisfying $l_\star(z, z_\mathrm{r}, u_\mathrm{r})<\tau$, we have
    \begin{equation}\label{eq:cost_controllable}
        J_N^\star(z, z_\mathrm{r}, u_\mathrm{r}) \leq \eta \ell_\star(z, z_\mathrm{r}, u_\mathrm{r}).
    \end{equation}
    Furthermore, let the stage cost be formulated as in \eqref{eq:stage_cost_tracking} with $Q,R\succ 0$ and let the MPC problem \eqref{eq:tracking_ocp} be feasible at time instance $k=0$. Then, for any $\bar{J}_N>0$, there exists a horizon length $\bar{N}\in\mathbb{Z}_+$ such that for every $N>\bar{N}$, and any $z_k$ satisfying $J_{N}^\star(z_k, z_\mathrm{r}, u_\mathrm{r})\leq\bar{J}_N$ and $r_k^\star\in\breve{\mathbb{S}}$, the MPC scheme \eqref{eq:tracking_ocp} with a constant reference $p^\mathbf{y}_\mathrm{ref}$ is recursively feasible, the projection of the set $\mathbb{S}_\mathrm{o}$ to the space $z$ is locally exponentially stable, and $\lim_{k\to\infty}\bar{H}_{N}^\star(z_k) = 0$, where
    \begin{multline}
        \bar{H}_{N}^\star(z_k) \triangleq H_{N}^\star(z_k) - V_\mathrm{r}^\mathrm{min} =\\ J_{N}^\star(z_k, z_\mathrm{r}^\star, u_\mathrm{r}^\star) + V_\mathrm{r}(z_\mathrm{r}^\star, u_\mathrm{r}^\star) - V_\mathrm{r}^\mathrm{min}.
    \end{multline}
\end{theorem}
\begin{proof}
    Using~\eqref{eq:cost_controllable}, $J_{N}^\star(z_k, z_\mathrm{r}, u_\mathrm{r})\leq \bar{J}_N$ implies that $J_{N}^\star(z_k, z_\mathrm{r}, u_\mathrm{r}) \leq \bar{\eta} \ell_\star(z_k, z_\mathrm{r}, u_\mathrm{r})$, with $\bar{\eta}=\max\{\eta, \bar{J}_N/\tau\}.$ The value of $\bar{J}_N$ can be selected w.l.o.g. to satisfy $\bar{J}_N \geq \eta \tau$. Then, choose $\bar{N} = \max\{\eta, \bar{\eta} (\eta -1)\}$. The stage cost \eqref{eq:stage_cost_tracking} with $Q,R\succ 0$ satisfy Assump.~1 in \cite{soloperto_nonlinear_2023}, then the proof of \cite[Prop.~3]{soloperto_nonlinear_2023} can be directly applied. The arguments there invoke Condition~\ref{cond:curvatures} at $r_k^\star$ and at the perturbed set-points $\hat{r}$, which lie on $c^\ast$, hence remain in $\breve{\mathbb{S}}$.
\end{proof}

The condition in \eqref{eq:cost_controllable} corresponds to a cost controllability assumption, which is commonly used to establish stability without terminal ingredients~\cite{boccia_stability_2014}. If a quadratically bounded terminal cost can be computed offline for all $r \in \mathbb{S}$ (e.g., following Sect.~\ref{sec:terminal_cond}), then \eqref{eq:cost_controllable} is automatically satisfied and thus serves only as a technical condition~\cite{limon_nonlinear_2018}. Theorem~\ref{thm:tracking_stability} then implies that, for a given $\bar{J}_N$, a horizon length $N > \bar{N}$ can be selected such that stability is guaranteed if $J_{N}^\star(z_k, z_\mathrm{r}, u_\mathrm{r}) < \bar{J}_N$ along the closed-loop trajectory, which is ensured by \eqref{eq:stability_constr}. Under special conditions, this constraint can be omitted, and stability still holds; see~\cite{soloperto_nonlinear_2023}. Furthermore, Theorem~\ref{thm:tracking_stability} establishes stability for a constant reference; however, as $p_\mathrm{ref}^\mathbf{y}(k)$ only enters the objective of \eqref{eq:tracking_ocp}, recursive feasibility is preserved under arbitrary reference changes. Consequently, $\bar{H}_N^\star$ may increase at the time instances of the changes but decreases in between, and once the reference remains constant, Theorem~\ref{thm:tracking_stability} applies from that time instance onward. This corresponds to the guarantees of standard artificial reference-based tracking MPC schemes~\cite{limon_nonlinear_2018}. Stability under dynamic references requires further analysis; see, e.g.,~\cite{kohler_nonlinear_2020b}, which is left for future work. Finally, the guarantees of Theorem~\ref{thm:tracking_stability} are local in the set-point space, i.e., the optimal artificial references are required to remain in $\breve{\mathbb{S}}$. This is not restrictive, as $r_k^\star$ minimizes the cost and thus naturally stays close to $\mathbb{S}_\mathrm{o}$. If a global guarantee is required, the condition can be enforced explicitly by changing \eqref{eq:setpoint_constr} to $r\in\breve{\mathbb{S}}$.

\section{Numerical example}\label{sec:num_examples}

\subsection{MSS model identification}
We identify and control a stochastic system that the following relations describe:
\begin{subequations}\label{eq:example_sys}
\begin{align}
    \mathbf{x}^{(1)}_{k+1} &= \big(0.2 + 0.8\,e^{-(\mathbf{x}_k^{(2)}+\mathbf{v}_k)^2}\big)\mathbf{x}_k^{(1)} + 0.3\sin(\mathbf{x}_k^{(2)})\,u_k,\\
    \mathbf{x}^{(2)}_{k+1} &= -0.4\,\mathbf{x}_k^{(1)} + \big(0.7 + 0.3\sin\mathbf{w}_k\big)\mathbf{x}_k^{(2)},
\end{align}
\end{subequations}
where the states $\mathbf{x}_k^{(1)}, \mathbf{x}_k^{(2)}$ are random variables taking values in $\mathbb{R}$, $u_k\in\mathbb{R}$ is the control input, and $\mathbf{v}_k, \mathbf{w}_k$ are i.i.d. noise processes with $\mathbf{v}_k\sim\mathcal{U}(-0.1, 0.1)$, $\mathbf{w}_k\sim\mathcal{U}(-\pi, \pi)$. The output equation is given as $\mathbf{y}_k = \mathbf{x}_k^{(1)}$, where $\mathbf{y}_k$ is the measured output taking values in $\mathbb{R}$. The initial state distributions are given as $\mathbf{x}^{(1)}_0, \mathbf{x}_0^{(2)}\sim\mathcal{U}(0,1)$. For model training, $u_k$ is sampled from a uniform distribution, as $u_k\sim\mathcal{U}(0, 5)$ with a sequence length of $N=8000$. The same input sequence has been applied to the system 10 times, and all 10 measurement series have been utilized during training. The settings associated with the model structure are summarized in Table~\ref{tab:hyperparams_example}. To prevent overfitting, $\ell_2$ regularization of the parameters~$\theta$ has been applied during training with a coefficient of~$10^{-6}$. The model is trained for 2000 epochs with the Adam optimizer, then 2000 epochs with the L-BFGS-B method following the identification pipeline adapted from \cite{beintema_meta-statespace_2024,bemporad_l-bfgs-b_2025}. The training time was roughly 15 minutes on a laptop with an Intel Core i7-13700HX processor and 16 GB of RAM. Note that, in the remainder, all reported computational times are associated with this machine.

\begin{table}
    \centering
    \scriptsize
    \renewcommand{\arraystretch}{1.25}
    \caption{Model structure for the MSS identification of \eqref{eq:example_sys}.}
    \begin{tabular}{lc}
        \hline
        Parameter name & Value\\
        \hline
        $\psi_\theta$ ANN struct. & 2 hidden layers, 32 neurons\\
        $f_\theta^\mathrm{z}$ ANN struct. & 2 hidden layers, 8 neurons\\
        $w_\theta$, $\mu_\theta$, $\Sigma_\theta$ ANN struct. & 2 hidden layers, 32 neurons\\
        Activation fun. (all ANNs) & hyperbolic tangent (tanh)\\
        Model order $n_\mathrm{z}$ & 3\\
        Gaussian components $n_\mathrm{G}$ & 12\\
        Encoder lag $n$ & 15\\
        \hline
    \end{tabular}
    \label{tab:hyperparams_example}
\end{table}

For testing the identified model, we generate a single input sequence with length $N_\mathrm{test}$ sampled from the same distribution as the training set, i.e., $u_k\sim\mathcal{U}(0,5)$, and apply the generated input to the data-generating system $S=1000$ times. Then, we evaluate the identified MSS model on all generated sequences and compute the average log-likelihood, as
\begin{equation}\label{eq:mean_log_likelihood}
    \frac{1}{N_\mathrm{test} S} \sum_{k=0}^{N_\mathrm{test}-1}\sum_{i=1}^S \log \hat{p}_k^\mathbf{y}(y_k^{(i)}),
\end{equation}
where $\hat{p}_k^\mathbf{y}(y)$ denotes the predicted output pdf by $h^\mathrm{z}_\theta$, while $y_k^{(i)}$ represents the output realization measured from the system at time step $k$ in the $i^\mathrm{th}$ sequence. With such a large $S$, the computed mean log-likelihood accurately represents the model quality. As shown in \cite{beintema_meta-statespace_2024}, the upper limit of \eqref{eq:mean_log_likelihood} can be given by
\begin{equation}\label{eq:log_like_upper_lim}
    -\frac{1}{N_\mathrm{test}} \sum_{k=0}^{N_\mathrm{test}-1} H(\{y_k^{(i)}\}_{i=1}^S),
\end{equation}
where $H$ is the differential entropy, which can be estimated from samples as proposed in \cite{vasicek_test_1976}. The formulation of \eqref{eq:log_like_upper_lim} offers an estimate for the upper limit of \eqref{eq:mean_log_likelihood}. The length of the test sequence is varied as $N_\mathrm{test}=n+\bar{N}_\mathrm{test}$ to evaluate model performance with different $\bar{N}_\mathrm{test}$-step-ahead predictions. The results are shown in Table~\ref{tab:log-likelihoods}. As visible, the model provides log-likelihood values close to the theoretical upper limit for all prediction horizons, which is remarkable for the considered highly nonlinear and stochastic data-generating system. The log-likelihood metric in \eqref{eq:mean_log_likelihood} and the upper bound \eqref{eq:log_like_upper_lim} can be formulated in the per-step fashion as well to investigate the accuracy of the model in each time step. The largest discrepancy between the per-step log-likelihood and the upper bound is only 0.15, indicating highly accurate simulation capabilities. We also illustrate this by comparing the sampled output pdfs to the MSS model prediction for certain time instances in Fig.~\ref{fig:example_sys_pdfs} (for the $\bar{N}_\mathrm{test}=50$ case). This result highlights that the applied MSS learning approach can effectively capture the stochastic elements in the data-generating system while applying a deterministic meta-state transition function.

\begin{figure*}
    \centering
    \includegraphics[scale=0.9]{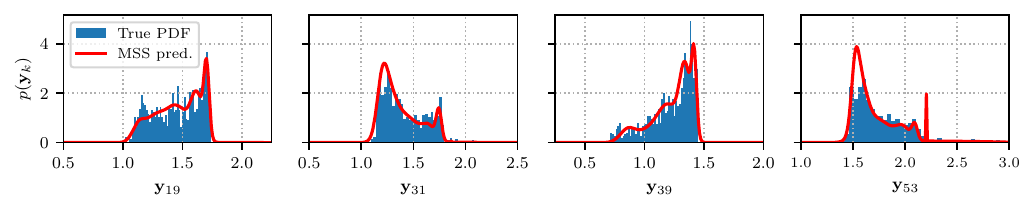}\vspace{-12pt}
    \caption{True and modeled output pdfs on the test data set (with a prediction horizon of 50) at randomly selected time instances.}\vspace{-4pt}
    \label{fig:example_sys_pdfs}
\end{figure*}

\begin{table}
    \centering
    \scriptsize
    \renewcommand{\arraystretch}{1.25}
    \caption{Average log-likelihood of identified MSS models with different prediction horizons and the estimated upper limit.}
    \begin{tabular}{lccccc}
        \hline
        $\bar{N}_\mathrm{test}$ & 5 & 10 & 25 & 50 & 75\\
        \hline
        Est. upper limit & 0.91 & 0.91 & 0.92 & 0.92 & 0.93\\
        Est. log-likelihood & 0.81 & 0.83 & 0.85 & 0.85 & 0.86\\
        \hline
    \end{tabular}
    \label{tab:log-likelihoods}
\end{table}

\subsection{Case 1: Constant output pdf reference}
To show the capabilities of the proposed MSS-SMPC approach\footnote{The code is available at \url{https://github.com/AIMotionLab-SZTAKI/meta-state-space-mpc}}, first, we search for an appropriate equilibrium point. For demonstration purposes, we aim to generate an output pdf that produces two modes by applying the following desired output pdf: $p_\mathrm{des}^\mathbf{y} = 0.5\cdot \mathcal{N}(2.5, 0.1^2) +  0.5\cdot \mathcal{N}(2.9, 0.1^2)$, then we perform a set-point search as discussed in Section~\ref{sec:equilibrium}. The KL divergence is approximated by the discussed Monte Carlo sampling method with $M=500$ samples. The desired pdf and the pdf at the computed equilibrium, evaluated at the sample points, are shown in Fig.~\ref{fig:PDF_search}. In the equilibrium search, we impose chance constraints $y_\mathrm{max}=3.5$, $p_\mathrm{max}=0.9$, and input bounds $u_\mathrm{min}=0$, $u_\mathrm{max}=5$, which are also used in the MPC formulation. A prediction horizon of $N=25$ is chosen, with tuned weights $Q=\mathrm{diag}(10,25,30)$ and $R=0.25$. The construction of the terminal controller follows Section~\ref{sec:terminal_cond}, and the terminal region is obtained via the sampling-based method discussed in Appendix~\ref{appendix:terminal_ingredients}. The resulting level-set ($\gamma=0.19$) satisfies Condition~\ref{assum:terminal_region}, and the $N$-step backward reachable set is verified to cover a sufficiently large region of the meta-state space by using the sampling-based algorithm in Section~\ref{sec:reachability}. The set-point searching algorithm requires 4 seconds, while the terminal controller construction and the sampling-based verification are nearly instantaneous, taking about 1.8 and 200 ms, respectively. The latter needs to be repeated until a suitable $\gamma$ value is found, which typically takes a few seconds. The reachable set verification method takes 12 minutes. We numerically verify the NLP regularity conditions discussed in Section~\ref{sec:reachability} that are required for the reachable set certification, which takes another 3 minutes.

First, to demonstrate the nominal stability/performance properties of the MSS-MPC scheme, we evaluate the method using the identified MSS model as the plant. The initial meta-state is set as $z_0=1.25\bar{z}$, and the true meta-state values are provided to the controller at each time instance (as if they were measured). The results in Fig.~\ref{fig:example_mpc_results} (displayed with purple) show that the steady-state output pdf matches the provided reference and chance constraints are satisfied throughout a 500-step run.

Then, the MPC is evaluated on the true stochastic system \eqref{eq:example_sys} over 1000 Monte Carlo runs, each with a 500-step-long closed-loop simulation, where now the learned encoder network is used to estimate the initial meta-states. The results are shown in Fig.~\ref{fig:example_mpc_results} with blue. After removing the transients of each simulation using a change-point detection procedure based on the \emph{Maximum Mean Discrepancy}~(MMD)~\cite{gretton_kernel_2012}, steady-state samples are aggregated across all realizations to estimate the steady-state output pdf, shown in Fig.~\ref{fig:example_mpc_results}. The resulting distribution matches the reference, which demonstrates unprecedented behavior among SMPC approaches, namely that complete output pdfs can be provided as a reference to control general nonlinear stochastic systems at real-time. The MSS chance constraint predictions and the empirical constraint satisfaction probability are close to each other, both well within the specified limit. The average computational time of the MPC is only 18 ms, with a worst-case solve time of 35 ms.

\begin{figure}
    \centering
    \includegraphics[scale=0.9]{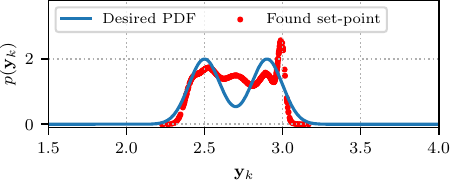}\vspace{-6pt}
    \caption{Desired output pdf compared to the one produced by the found equilibrium point at the sample points.}\vspace{-4pt}
    \label{fig:PDF_search}
\end{figure}

\begin{figure*}
    \centering
    \includegraphics[scale=0.9]{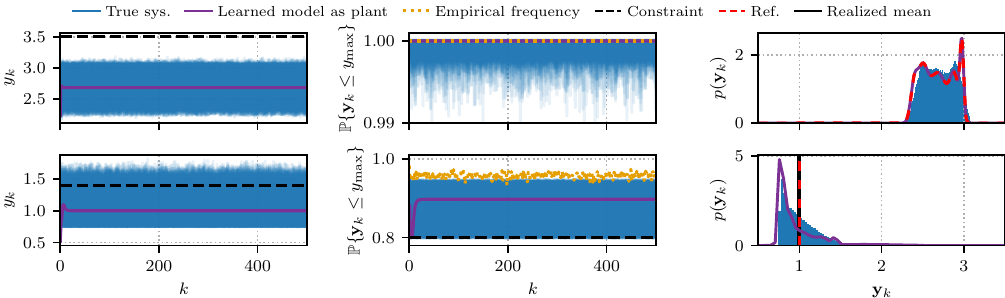}\vspace{-6pt}
    \caption[dummy.]{Results of the MSS-MPC scheme applied to both the true system (\blueline) and the model as the plant (\purpleline) for Case 1 (upper row) and Case 2 (lower panels). Left panels show the sampled outputs of the system \eqref{eq:example_sys} for all 1000 Monte Carlo runs and the output constraint (\blackdashedline). For the case when the learned model acts as the plant, only the mean is displayed   (\purpleline). Middle panels show the chance constraint values predicted by the learned model along each run, the empirical constraint satisfaction frequency of the true system over the 100 runs (\orangedottedline), the exact constraint satisfaction probability with the learned model as the plant (\purpleline), and the specified probability bound (\blackdashedline). In Case 1, both the empirical frequency and the exact probability equal 1 for all $k$. Right panels show the estimated steady-state output distribution using all 1000 Monte Carlo runs (after discarding the detected transients) and the steady-state output pdf of the learned model-based closed-loop run (\purpleline). The red dashed line (\reddashedline) represents the designed reference output pdf for Case 1, and the reference output mean for Case 2. For the latter, the realized mean is also shown (\blackline).}\vspace{-4pt}
    \label{fig:example_mpc_results}
\end{figure*}

\subsection{Case 2: Constant output mean reference}
As a second example, we compute an equilibrium that generates $\mathbb{E}\{\mathbf{y}_k\}=1$ by solving \eqref{eqs:optim_set_point_search} with $\beta=0$, focusing on mean-matching. We use the same input bounds as in Case 1, while modifying the chance constraints to $y_\mathrm{max}=1.4$ and $p_\mathrm{max}=0.8$. The weights are set as $Q=\mathrm{diag}(1.5, 15, 2.5)$, $R=1$. The terminal ingredients are constructed as before, and validated for $\gamma=0.001$. We also illustrate the use of multiple set-points to enlarge the domain of attraction by designing two alternative equilibria that both approximately achieve the desired mean. Their $N$-step backward reachable sets are computed as discussed in Section~\ref{sec:reachability}, and are shown in Fig.~\ref{fig:DOA}. Similarly to Case 1, the set-point search, terminal controller synthesis, and terminal set verification are rapid, taking about 300, 10, and 300 ms, respectively. The reachable set computation (with checking the regularity conditions) takes 33 minutes per set-point.

Again, we evaluate the MSS-MPC approach first using the identified model as the plant, under the same conditions as in Case~1. Fig.~\ref{fig:example_mpc_results} shows that the expected output value quickly converges to the reference, while the chance constraints are never violated over the 500-step simulation. Then, the MPC is evaluated using the true system over 1000 Monte Carlo runs, also shown in Fig.~\ref{fig:example_mpc_results}. Due to the (intentionally) lower $p_\mathrm{max}$, occasional output constraint violations occur, though the predicted chance constraints by the model remain satisfied. Across the 1000 Monte Carlo runs, the aggregated constraint satisfaction probability is 95.78\%, while the per-step values (shown in Fig.~\ref{fig:example_mpc_results}) satisfy the specified bound at each time step. Slight model mismatch is noticeable regarding the chance constraints, as along the plant trajectories, model predictions drop to 0.8, even though the realized satisfaction is higher. This is caused by the distribution shift of the input signal between the training data and the closed-loop system. However, this mismatch makes the controller conservative rather than unsafe and could be reduced by a dedicated experiment design in the model learning phase. Nevertheless, the achieved steady-state output mean accurately matches the reference.

For benchmarking, we consider two GP-MPC schemes: a naive formulation that includes the full state covariance matrices as decision variables, and the computationally efficient zoGP-MPC approach~\cite{lahr_zero-order_2023}, which employs zero-order uncertainty propagation. For implementation simplicity, both methods assume full-state measurements, unlike MSS-MPC, which relies only on input-output data. The GP model is trained using the same data-generation procedure as for the MSS identification. Due to the large dataset, a sparse GP with 50 inducing points is employed. For both GP-MPC schemes, the input and state references are set to the equilibrium corresponding to the MSS-MPC scheme, with the reference for $\mathbf{x}_k^{(2)}$ set to zero, while the weights $Q$ and $R$ are tuned manually. To ensure a fair comparison, the naive GP-MPC is implemented using CasADi, similarly to MSS-MPC. The results, summarized in Table~\ref{tab:gp_mpc_comparison}, show that the naive GP-MPC is computationally demanding, and the sample output mean deviates significantly from the reference, due to the excessive conservatism of linearization-based uncertainty propagation in the chance constraint evaluation. In contrast, zoGP-MPC achieves the lowest computational cost while providing reasonable control performance. Although the primary advantage of the MSS-MPC method is that full output pdfs can be provided as reference, it still achieves superior accuracy in the current mean-matching control task while remaining within real-time computational limits. Note that, similarly to the applied zoGP-MPC implementation, the MSS-MPC scheme could also be implemented using the \texttt{acados} library~\cite{verschueren_acadosmodular_2022}, e.g., through the \texttt{L4acados} toolbox~\cite{lahr_l4acados_2026}, to exploit efficient embedded SQP solvers and decrease computation time. However, such software developments are beyond the scope of this work. Note also that (for implementation simplicity) neither GP-MPC formulation includes terminal constraints. Incorporating these might increase computational cost.

\begin{table}
    \centering
    \scriptsize
    \renewcommand{\arraystretch}{1.25}
    \caption{Comparing the MSS-MPC to GP-MPC schemes on Case 2.}
    \label{tab:gp_mpc_comparison}
    \begin{tabular}{lcc}
    \hline
        Method & Avg. / max. comp. time & $1/N \sum\vert y_k-y_\mathrm{ref}\vert$\\
    \hline
        GP-MPC & $121\pm 31$ ms/ 519 ms& 0.1733\\
        zoGP-MPC & $\mathbf{7\pm 7}$ \textbf{ms} / 111 ms & 0.0062\\
        MSS-MPC & $49\pm 5$ ms / \textbf{69 ms}& \bf 0.0018\\
    \hline
    \end{tabular}\vspace{-4pt}
\end{table}

\begin{figure}
    \centering
    \includegraphics{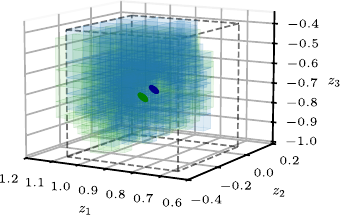}\vspace{-4pt}
    \caption{The $N$-step backward reachable sets as a union of semi-transparent cuboids, associated with the two set-points in Case 2. The terminal regions are also shown as opaque ellipsoids. The dashed lines represent the region of interest.}\vspace{-4pt}
    \label{fig:DOA}
\end{figure}

\subsection{Case 3: Changing output pdf reference}
Next, we demonstrate the capabilities of the tracking SMPC scheme from Section~\ref{sec:changing_ref} using two simulation case studies. First, we prescribe a time-varying reference for the full output pdf: a gamma distribution for the first 1000 steps, then it is suddenly changed to a beta distribution for the next 1000 steps, as shown in Fig.~\ref{fig:pdf_tracking-example}. Note that exactly achieving these reference distributions might not be possible; the objective is to minimize the KL divergence between the achieved and reference pdfs. For simplicity, in this example, only input bounds $0 \leq u_k \leq 5$ are enforced. The MPC is implemented with $N=60$, $\bar{J}_N=0.15$, $Q=I_3$, and $R=1$, and a regularized offset cost with $\lambda=10^{-6}$ and $r_\mathrm{c}$ corresponding to the minimum-norm feasible set-point.\footnote{Condition \eqref{eq:cost_controllable} is not verified, but $N$ and $\bar{J}_N$ are chosen by trial and error, while Condition~\ref{cond:curvatures} is verified numerically. As $n_\mathrm{u}=1$, $\mathbb{S}$ is a curve parametrized by $u$, along which $d_{\mathbb{S}_\mathrm{o}}$ reduces to the arc length. Evaluations on a fine grid of set-points show satisfaction of \eqref{eq:curvature_assum_1}--\eqref{eq:curvature_assum_2}.} The KL divergence approximation is implemented as \eqref{eq:KL_approx} with $M=100$ grid points. Samples are only generated once for each reference pdf, and from then on, the same grid points are reused. For each reference, the output pdfs are estimated after discarding the first 50 samples; the remaining samples were verified to represent steady-state distributions using the MMD-based criterion introduced in Case~1. Results in Fig.~\ref{fig:pdf_tracking-example} show that the approach tracks both reference distributions closely and, remarkably, it can handle the sudden change of the pdf reference signal. The average computational time for the MPC is 37 ms with a worst-case solve time of 177 ms, although, as mentioned, chance constraints are not applied.

\subsection{Case 4: Changing output mean reference}
Finally, the mean-matching offset cost is applied with a time-varying mean reference. All MPC settings are as in Case 3, but now $r_\mathrm{c}$ is set to realize $\mathbb{E}\{\mathbf{y}_k\}\approx 2$, corresponding to the middle of the operating domain, and we additionally apply a chance constraint $\mathbb{P}\{\mathbf{y}_k \leq 3\} \geq 0.9$. As shown in Fig.~\ref{fig:mean_tracking_example}, the method tracks the reference mean accurately. Between $k=800,\dots, 1000$, the reference is infeasible due to the output (chance) constraint; in this case, the MPC drives the system as close to the reference as possible, by adhering to the constraints. This highlights the flexibility and efficiency of the proposed approach. The average MPC computational time is 124 ms, while the maximum time is 796~ms, which remains efficient despite using a naive implementation and optimization scheme. In practice, the outlying worst-case solver time can be significantly reduced by using real-time iterations~\cite{diehl_real-time_2005}, e.g., following the highly efficient embedded implementation in \texttt{acados}~\cite{verschueren_acadosmodular_2022}.

\begin{figure}
    \centering
    \includegraphics[scale=0.9]{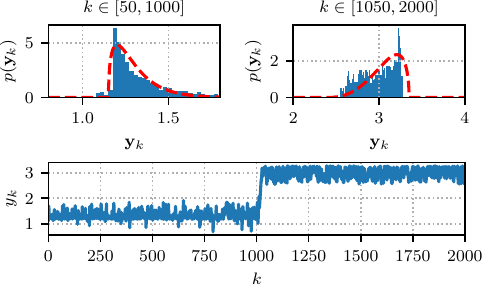}\vspace{-6pt}
    \caption[Results of the set-point tracking SMPC approach with Gaussian output pdf references. The upper panels show the sample distribution of the output for each reference (after discarding transients) compared to the control objective. The lower panel shows the resulting output signal.]{Results for Case 3. The upper panels show the sample distribution of the output for each reference (after discarding transients) compared to the control objective (\reddashedline). The lower panel shows the resulting output signal.}\vspace{-2pt}
    \label{fig:pdf_tracking-example}
\end{figure}

\begin{figure}
    \centering
    \includegraphics[scale=0.9]{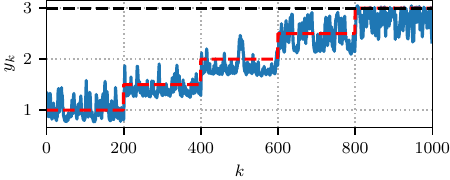}\vspace{-6pt}
    \caption[Results of the mean-matching SMPC approach with online changing reference. The black dashed line represents the output constraint, while the red dashed line corresponds to $y_\mathrm{ref}(k)$.]{Results of Case 4. The black dashed line (\blackdashedline) represents the output constraint, while the red dashed line (\reddashedline) corresponds to the output mean reference.}\vspace{-4pt}
    \label{fig:mean_tracking_example}
\end{figure}

\section{Conclusion}\label{sec:conclusion}
In this paper, we have proposed a novel stochastic model-predictive control scheme based on the meta-state-space representation. Utilizing an MSS model, the proposed formulation overcomes various drawbacks of conventional SMPC approaches. The deterministic meta-state transition function realizes a computationally tractable method for uncertainty propagation, and the applied Gaussian Mixture Model parametrization of the output pdf enables the efficient evaluation of the chance constraints. The presented method can be efficiently applied to general nonlinear stochastic systems, for both constant and online-changing set-point stabilization,  without requiring linearization or Gaussian uncertainty approximation, which are typically present in other SMPC approaches. Finally, we have demonstrated the effectiveness of the proposed methods via an extensive simulation example.

\section*{Acknowledgments}
We thank Gerben Beintema for his initial ideas that contributed to this paper. We also thank Kristóf Floch for helping with the GP-MPC implementations.

\bibliographystyle{plain}
\bibliography{references}

@article{beintema_meta-statespace_2024,
	title = {Meta-state-space learning: {An} identification approach for stochastic dynamical systems},
	volume = {167},
	shorttitle = {Meta-state–space learning},
	journal = {Automatica},
	author = {Beintema, Gerben I. and Schoukens, Maarten and T\'oth, Roland},
	year = {2024},
	pages = {111787}
}

@Article{andersson_casadi_2019,
  author = {Joel A E Andersson and Joris Gillis and Greg Horn
            and James B Rawlings and Moritz Diehl},
  title = {{CasADi} -- {A} software framework for nonlinear optimization and optimal control},
  journal = {Mathematical Programming Computation},
  volume = {11},
  number = {1},
  pages = {1--36},
  year = {2019}
}

@article{bemporad_l-bfgs-b_2025,
	title = {An {L}-{BFGS}-{B} {Approach} for {Linear} and {Nonlinear} {System} {Identification} {Under} $\ell_1$ and {Group}-{Lasso} {Regularization}},
	volume = {70},
	number = {7},
	journal = {IEEE Transactions on Automatic Control},
	author = {Bemporad, Alberto},
	year = {2025},
	pages = {4857--4864}
}

@article{schwarm_chance-constrained_1999,
	title = {Chance-constrained model predictive control},
	volume = {45},
	number = {8},
	journal = {AIChE Journal},
	author = {Schwarm, Alexander T. and Nikolaou, Michael},
	year = {1999},
	pages = {1743--1752}
}

@article{van_hessem_stochastic_2006,
	title = {Stochastic closed-loop model predictive control of continuous nonlinear chemical processes},
	volume = {16},
	number = {3},
	journal = {Journal of Process Control},
	author = {Van Hessem, Dennis and Bosgra, Okko},
	year = {2006},
	pages = {225--241}
}

@inproceedings{herceg_scenario-based_2025,
	title = {A {Scenario}-{Based} {Model} {Predictive} {Control} {Scheme} for {Pandemic} {Response} {Through} {Non}-{Pharmaceutical} {Interventions}},
	booktitle = {Proc. of the {IEEE} {Conference} on {Control} {Technology} and {Applications}},
	author = {Herceg, Domagoj and Dell'Oro, Marco and Bertollo, Riccardo and Miura, Fuminari and de Klaver, Paul and Breschi, Valentina and Krishnamoorthy, Dinesh and Salazar, Mauro},
	year = {2025},
	pages = {139--144}
}

@inproceedings{liu_stochastic_2015,
	title = {Stochastic predictive control for lane keeping assistance systems using a linear time-varying model},
	booktitle = {Proc. of the {American} {Control} {Conference}},
	author = {Liu, Changchun and Carvalho, Ashwin and Schildbach, Georg and Hedrick, J. Karl},
	year = {2015},
	pages = {3355--3360}
}

@article{mesbah_stochastic_2016,
	title = {Stochastic {Model} {Predictive} {Control}: {An} {Overview} and {Perspectives} for {Future} {Research}},
	volume = {36},
	shorttitle = {Stochastic {Model} {Predictive} {Control}},
	number = {6},
	journal = {IEEE Control Systems Magazine},
	author = {Mesbah, Ali},
	year = {2016},
	pages = {30--44}
}

@article{cannon_stochastic_2012,
	title = {Stochastic tube {MPC} with state estimation},
	volume = {48},
	number = {3},
	journal = {Automatica},
	author = {Cannon, Mark and Cheng, Qifeng and Kouvaritakis, Basil and Raković, Saša V.},
	year = {2012},
	pages = {536--541}
}

@phdthesis{beintema_datadriven_2024,
	type = {Phd {Thesis}},
	title = {Data–driven {Learning} of {Nonlinear} {Dynamic} {Systems}: {A} {Deep} {Neural} {State}–{Space} {Approach}},
	shorttitle = {Data–driven {Learning} of {Nonlinear} {Dynamic} {Systems}},
	school = {Eindhoven University of Technology},
	author = {Beintema, Gerben I.},
	year = {2024}
}

@article{chen_quasi-infinite_1998,
	title = {A {Quasi}-{Infinite} {Horizon} {Nonlinear} {Model} {Predictive} {Control} {Scheme} with {Guaranteed} {Stability}},
	volume = {34},
	number = {10},
	journal = {Automatica},
	author = {Chen, Hong and Allgöwer, Frank},
	year = {1998},
	pages = {1205--1217}
}

@book{rawlings_model_2017,
	edition = {2nd},
	title = {Model {Predictive} {Control}: {Theory}, {Computation}, and {Design}},
	publisher = {Nob Hill Publishing},
	author = {Rawlings, James B. and Mayne, David Q. and Diehl, Moritz},
	year = {2017},
}

@book{grune_nonlinear_2017,
    title = {Nonlinear {Model} {Predictive} {Control}},
    series = {Communications and {Control} {Engineering}},
    publisher = {Springer International Publishing},
    author = {Grüne, Lars and Pannek, Jürgen},
    year = {2017}
}

@ARTICLE{limon_nonlinear_2018,
    author={Limon, Daniel and Ferramosca, Antonio and Alvarado, Ignacio and Alamo, Teodoro},
    journal={IEEE Transactions on Automatic Control}, 
    title={Nonlinear MPC for Tracking Piece-Wise Constant Reference Signals}, 
    year={2018},
    volume={63},
    number={11},
    pages={3735-3750}
}

@article{vasicek_test_1976,
    author = {Vasicek, Oldrich},
    title = {A Test for Normality Based on Sample Entropy},
    journal = {Journal of the Royal Statistical Society: Series B (Methodological)},
    volume = {38},
    number = {1},
    pages = {54-59},
    year = {1976}
}

@book{paul_stochastic_2013,
    address = {Heidelberg},
    edition = {2nd},
    title = {Stochastic {Processes}: {From} {Physics} to {Finance}},
    shorttitle = {Stochastic {Processes}},
    publisher = {Springer International Publishing},
    author = {Paul, Wolfgang and Baschnagel, Jörg},
    year = {2013}
}

@book{goodfellow_deep_2016,
    title={Deep Learning},
    author={Ian Goodfellow and Yoshua Bengio and Aaron Courville},
    publisher={MIT Press},
    year={2016}
}

@techreport{bishop_mixture_1994,
    title = {Mixture {Density} {Networks}},
    institution = {Neural Computing Research Group, Aston University},
    author = {Bishop, Christopher M.},
    year = {1994},
}

@inproceedings{cui_comparison_2015,
  title={Comparison of Kullback-Leibler divergence approximation methods between Gaussian mixture models for satellite image retrieval},
  author={Cui, Shiyong and Datcu, Mihai},
  booktitle={Proc. of the IEEE International Geoscience and Remote Sensing Symposium},
  pages={3719--3722},
  year={2015}
}

@INPROCEEDINGS{hershey_approximating_2007,
  author={Hershey, John R. and Olsen, Peder A.},
  booktitle={Proc. of the IEEE International Conference on Acoustics, Speech and Signal Processing}, 
  title={Approximating the Kullback Leibler Divergence Between Gaussian Mixture Models}, 
  year={2007},
  pages={IV-317--IV-320},
}

@article{gen_numerical_1992,
    author = {Alan Genz},
    title = {Numerical Computation of Multivariate Normal Probabilities},
    journal = {Journal of Computational and Graphical Statistics},
    volume = {1},
    number = {2},
    pages = {141--149},
    year = {1992}
}

@article{kohler_nonlinear_2020,
    title = {A {Nonlinear} {Model} {Predictive} {Control} {Framework} {Using} {Reference} {Generic} {Terminal} {Ingredients}},
    volume = {65},
    number = {8},
    journal = {IEEE Transactions on Automatic Control},
    author = {Köhler, Johannes and Müller, Matthias A. and Allgöwer, Frank},
    year = {2020},
    pages = {3576--3583}
}

@article{conte_distributed_2016,
    title = {Distributed synthesis and stability of cooperative distributed model predictive control for linear systems},
    volume = {69},
    journal = {Automatica},
    author = {Conte, Christian and Jones, Colin N. and Morari, Manfred and Zeilinger, Melanie N.},
    year = {2016},
    pages = {117--125}
}

@phdthesis{kohler_analysis_2021,
    type = {Phd {Thesis}},
    title = {Analysis and design of {MPC} frameworks for dynamic operation of nonlinear constrained systems},
    school = {Universität Stuttgart},
    author = {Köhler, Johannes},
    year = {2021}
}

@article{nurbayeva_nonlinear_2025,
    title = {Nonlinear model predictive control with set terminal constraint for safe robot motion planning via speed and separation monitoring},
    volume = {154},
    journal = {Control Engineering Practice},
    author = {Nurbayeva, Aigerim and Rubagotti, Matteo},
    year = {2025},
    pages = {106155}
}

@article{kohler_nonlinear_2020b,
	title = {A nonlinear tracking model predictive control scheme for dynamic target signals},
	volume = {118},
	journal = {Automatica},
	author = {Köhler, Johannes and Müller, Matthias A. and Allgöwer, Frank},
	year = {2020},
	pages = {109030}
}

@article{boccia_stability_2014,
    title = {Stability and feasibility of state constrained {MPC} without stabilizing terminal constraints},
    volume = {72},
    journal = {Systems \& Control Letters},
    author = {Boccia, Andrea and Grüne, Lars and Worthmann, Karl},
    year = {2014},
    pages = {14--21}
}

@article{soloperto_nonlinear_2023,
    title = {A {Nonlinear} {MPC} {Scheme} for {Output} {Tracking} {Without} {Terminal} {Ingredients}},
    volume = {68},
    number = {4},
    journal = {IEEE Transactions on Automatic Control},
    author = {Soloperto, Raffaele and Köhler, Johannes and Allgöwer, Frank},
    year = {2023},
    pages = {2368--2375}
}

@inproceedings{pasqualetti_controllability_2014,
    title = {Controllability metrics, limitations and algorithms for complex networks},
    booktitle = {Proc. of the {American} {Control} {Conference}},
    author = {Pasqualetti, Fabio and Zampieri, Sandro and Bullo, Francesco},
    year = {2014},
    pages = {3287--3292}
}

@article{vandenberghe_determinant_1998,
    title = {Determinant {Maximization} with {Linear} {Matrix} {Inequality} {Constraints}},
    volume = {19},
    number = {2},
    journal = {SIAM Journal on Matrix Analysis and Applications},
    author = {Vandenberghe, Lieven and Boyd, Stephen and Wu, Shao-Po},
    year = {1998},
    pages = {499--533}
}

@article{costa_optimal_2012,
    title = {Optimal mean–variance control for discrete-time linear systems with {Markovian} jumps and multiplicative noises},
    volume = {48},
    number = {2},
    journal = {Automatica},
    author = {Costa, Oswaldo L. V. and Oliveira, Alexandre de},
    year = {2012},
    pages = {304--315}
}

@article{hakobyan_risk-aware_2019,
    title = {Risk-{Aware} {Motion} {Planning} and {Control} {Using} {CVaR}-{Constrained} {Optimization}},
    volume = {4},
    number = {4},
    journal = {IEEE Robotics and Automation Letters},
    author = {Hakobyan, Astghik and Kim, Gyeong Chan and Yang, Insoon},
    year = {2019},
    pages = {3924--3931}
}

@article{limon_enlarging_2005,
    title = {Enlarging the domain of attraction of {MPC} controllers},
    volume = {41},
    number = {4},
    journal = {Automatica},
    author = {Limon, Daniel and Alamo, Teodoro and Camacho, Eduardo F.},
    year = {2005},
    pages = {629--635}
}

@article{fagiano_generalized_2013,
    title = {Generalized terminal state constraint for model predictive control},
    volume = {49},
    number = {9},
    journal = {Automatica},
    author = {Fagiano, Lorenzo and Teel, Andrew R.},
    year = {2013},
    pages = {2622--2631}
}

@article{abdufattokhov_learning_2024,
    title = {Learning {Lyapunov} terminal costs from data for complexity reduction in nonlinear model predictive control},
    volume = {34},
    number = {13},
    journal = {International Journal of Robust and Nonlinear Control},
    author = {Abdufattokhov, Shokhjakhon and Zanon, Mario and Bemporad, Alberto},
    year = {2024},
    pages = {8676--8691}
}

@article{bobiti_automated-sampling-based_2018,
    title = {Automated-{Sampling}-{Based} {Stability} {Verification} and {DOA} {Estimation} for {Nonlinear} {Systems}},
    volume = {63},
    number = {11},
    journal = {IEEE Transactions on Automatic Control},
    author = {Bobiti, Ruxandra and Lazar, Mircea},
    year = {2018},
    pages = {3659--3674}
}

@inproceedings{kocijan_gaussian_2004,
    title = {Gaussian process model based predictive control},
    volume = {3},
    booktitle = {Proc. of the {American} {Control} {Conference}},
    author = {Kocijan, Juš and Murray-Smith, Roderick and Rasmussen, Carl E. and Girard, Agathe},
    year = {2004},
    pages = {2214--2219}
}

@article{hewing_cautious_2020,
    title = {Cautious {Model} {Predictive} {Control} {Using} {Gaussian} {Process} {Regression}},
    volume = {28},
    number = {6},
    journal = {IEEE Transactions on Control Systems Technology},
    author = {Hewing, Lukas and Kabzan, Juraj and Zeilinger, Melanie N.},
    year = {2020},
    pages = {2736--2743}
}

@article{gretton_kernel_2012,
    title = {A {Kernel} {Two}-{Sample} {Test}},
    volume = {13},
    number = {25},
    journal = {Journal of Machine Learning Research},
    author = {Gretton, Arthur and Borgwardt, Karsten M. and Rasch, Malte J. and Schölkopf, Bernhard and Smola, Alexander},
    year = {2012},
    pages = {723--773}
}

@article{lahr_l4acados_2026,
    title = {L4acados: {Learning}-{Based} {Models} for acados, {Applied} to {Gaussian} {Process}-{Based} {Predictive} {Control}},
    shorttitle = {L4acados},
    journal = {IEEE Transactions on Control Systems Technology},
    author = {Lahr, Amon and Näf, Joshua and Wabersich, Kim P. and Frey, Jonathan and Siehl, Pascal and Carron, Andrea and Diehl, Moritz and Zeilinger, Melanie N.},
    year={2026},
    volume={34},
    number={3},
    pages={1091-1105},
}

@article{lahr_zero-order_2023,
    title = {Zero-order optimization for {Gaussian} process-based model predictive control},
    volume = {74},
    journal = {European Journal of Control},
    author = {Lahr, Amon and Zanelli, Andrea and Carron, Andrea and Zeilinger, Melanie N.},
    year = {2023},
    pages = {100862}
}

@article{verschueren_acadosmodular_2022,
    title = {acados—a modular open-source framework for fast embedded optimal control},
    volume = {14},
    number = {1},
    journal = {Mathematical Programming Computation},
    author = {Verschueren, Robin and Frison, Gianluca and Kouzoupis, Dimitris and Frey, Jonathan and Duijkeren, Niels van and Zanelli, Andrea and Novoselnik, Branimir and Albin, Thivaharan and Quirynen, Rien and Diehl, Moritz},
    year = {2022},
    pages = {147--183}
}

@inproceedings{mesbah_stochastic_2014,
    title = {Stochastic nonlinear model predictive control with probabilistic constraints},
    booktitle = {Proc. of the {American} {Control} {Conference}},
    author = {Mesbah, Ali and Streif, Stefan and Findeisen, Rolf and Braatz, Richard D.},
    year = {2014},
    pages = {2413--2419},
}

@article{wang_robust_1999,
    title = {Robust control of the output probability density functions for multivariable stochastic systems with guaranteed stability},
    volume = {44},
    number = {11},
    journal = {IEEE Transactions on Automatic Control},
    author = {Wang, Hong},
    year = {1999},
    pages = {2103--2107},
}

@inproceedings{buehler_lyapunov-based_2016,
    title = {Lyapunov-based stochastic nonlinear model predictive control: {Shaping} the state probability distribution functions},
    booktitle = {Proc. of the {American} {Control} {Conference}},
    author = {Buehler, Edward A. and Paulson, Joel A. and Mesbah, Ali},
    year = {2016},
    pages = {5389--5394},
}

@article{karny_towards_1996,
    title = {Towards fully probabilistic control design},
    volume = {32},
    number = {12},
    journal = {Automatica},
    author = {Kárný, Miroslav},
    year = {1996},
    pages = {1719--1722},
}

@article{diehl_real-time_2005,
    title = {A {Real}-{Time} {Iteration} {Scheme} for {Nonlinear} {Optimization} in {Optimal} {Feedback} {Control}},
    volume = {43},
    number = {5},
    journal = {SIAM Journal on Control and Optimization},
    author = {Diehl, Moritz and Bock, Hans Georg and Schlöder, Johannes P.},
    year = {2005},
    pages = {1714--1736},
}

@inproceedings{kohler_simple_2019,
    title = {A simple framework for nonlinear robust output-feedback {MPC}},
    booktitle = {Proc. of the 18th {European} {Control} {Conference}},
    author = {Köhler, Johannes and Müller, Matthias A. and Allgöwer, Frank},
    year = {2019},
    pages = {793--798},
}

@book{bonnans2013perturbation,
  title={Perturbation Analysis of Optimization Problems},
  author={Bonnans, J. Frédéric and Shapiro, Alexander},
  series={Springer Series in Operations Research and Financial Engineering},
  year={2013},
  publisher={Springer New York}
}

\appendix
\section{Sampling-based design of the terminal set}\label{appendix:terminal_ingredients}
The largest level set of $V_\mathrm{f}(z)$ that satisfies the input constraints can be found by solving the following \emph{linear program} (LP), similarly as in \cite{conte_distributed_2016}:
\begin{subequations}\label{eq:u_LP_terminal}
\begin{align}
    \max_{\gamma} \quad&\gamma\\
    \text{s.t.} \quad&  \gamma \|P^{-1/2} K^\top H_{\mathrm{u},i}^\top \|_2^2 \leq \left(h_{\mathrm{u},i} - H_{\mathrm{u},i} u_\ast\right)^2, \, \forall i\in\mathbb{I}_1^{l_\mathrm{u}}.
\end{align}
\end{subequations}
Since \eqref{eq:u_LP_terminal} is an LP in only one variable, it can be efficiently solved numerically. The resulting $\gamma_\mathrm{u}$ level set ensures that $\kappa_\mathrm{f}(z)\in U$, for all $V_\mathrm{f}(z)\leq \gamma_\mathrm{u}$. Next, $\gamma_{g_\mathrm{min}}$, and $\gamma_{g_\mathrm{max}}$ can be both obtained by solving the following optimization problem adapted from \cite{nurbayeva_nonlinear_2025}:
\begin{subequations}\label{eq:optim_chance_constr}
\begin{align}
    \min_{z} \quad& (z-\bar{z})^\top P (z- \bar{z})\\
    \text{s.t.} \quad& g_\mathrm{min}(z, \kappa_\mathrm{z}(z)) \leq 0.
\end{align}
\end{subequations}
The optimization \eqref{eq:optim_chance_constr} yields a minimizer $z_{g_\mathrm{min}}$, the corresponding $\gamma$ bound can be computed as $\gamma_{g_\mathrm{min}} = (z_{g_\mathrm{min}}-\bar{z})^\top P (z_{g_\mathrm{min}}- \bar{z})$, then the same can be repeated for $g_\mathrm{max}$ to acquire $\gamma_{g_\mathrm{max}}$. Finally, obtaining a less conservative value for $\gamma_1$ compared to the analytical formulas in \cite{kohler_nonlinear_2020,rawlings_model_2017} is more complex. We adopt a similar method as in \cite{chen_quasi-infinite_1998}. First, select a candidate $\gamma=\min\{\gamma_\mathrm{u}, \gamma_{g_\mathrm{min}}, \gamma_{g_\mathrm{max}}\}$, then solve the following nonlinear optimization problem:
\begin{subequations}\label{eq:terminal_decrease_optim_problem}
\begin{align}
    \min_z \quad& V_\mathrm{f}(f^\mathrm{z}(z, \kappa_\mathrm{f}(z)) + \ell(z, \kappa_\mathrm{f}(z)) - V_\mathrm{f}(z)\label{eq:terminal_decrease_optim}\\
    \text{s.t.} \quad& (z-\bar{z})^\top P(z-\bar{z})\leq \gamma.
\end{align}
\end{subequations}
If the optimal value of \eqref{eq:terminal_decrease_optim} is positive, decrease $\gamma$ until a suitable value is found.
\begin{remark}
    The optimization problem \eqref{eq:terminal_decrease_optim_problem} is generally non-convex; thus, it might get stuck in a local minima, ultimately generating a $\gamma$ value that violates Condition~\ref{assum:terminal_region}. Alternatively, it is possible to generate random vectors or grid the space $\|\Delta z\|_P^2\leq \alpha$, then numerically check that \eqref{eq:terminal_cost_decrease} is satisfied for all samples (or grid points). This alternative approach only provides a sampling-based guarantee, but it might be preferred to overcome the challenges of solving \eqref{eq:terminal_decrease_optim_problem}.
\end{remark}

\section{Proof of Lemma~\ref{lem:terminal_ingredients}}\label{sec:appendix_lemma_terminal_ing_proof}
The first part of the proof closely follows Lem.~1 in~\cite{kohler_nonlinear_2020} (also in line with \cite{chen_quasi-infinite_1998,rawlings_model_2017}). Denote $\Delta z_k\triangleq z_k - \bar{z}$ and $\Delta u_k = K\Delta z_k$, thus:
\begin{equation}
    f^\mathrm{z}(z_k, \kappa_\mathrm{f}(z_k)) = f^\mathrm{z}(\bar{z}, \bar{u}) + A^\mathrm{z}\Delta z_k + B^\mathrm{z} \Delta u_k + \Delta f^\mathrm{z}(\Delta z_k),
\end{equation}
where $\Delta f^\mathrm{z}$ is the linearization error associated with $(A^\mathrm{z}, B^\mathrm{z})$. Condition~\ref{assum:f_cont} directly implies the existence of a $T\in\mathbb{R}_{\geq 0}$ such that $\|\Delta f^\mathrm{z}(\Delta z_k)\|_2 \leq T (1 + \|K\|_2^2) \|\Delta z_k\|_2^2$. Let the following notations be introduced:
\begin{align}
\lambda_\mathrm{max}^\mathrm{A_K} &\triangleq \lambda_\mathrm{max}((A^\mathrm{z}+B^\mathrm{z}K)^\top P (A^\mathrm{z}+B^\mathrm{z}K))\\
c_\mathrm{f} &\triangleq \tfrac{1}{\sqrt{\lambda_\mathrm{max}(P)}}(\sqrt{\lambda_\mathrm{max}^\mathrm{A_K} + \varepsilon} - \sqrt{\lambda_\mathrm{max}^\mathrm{A_K}}),\\
\gamma_1 &\triangleq \tfrac{\lambda_\mathrm{min}(P) c^2_\mathrm{f}}{T^2(1 + \lambda_\mathrm{max}^2(K))^2}.
\end{align}
Note that $Q, R\succ 0$ guarantees Condition~\ref{property:pos_def_stage_cost} to hold. Then, the proof of \cite[Lem.~1]{kohler_nonlinear_2020} directly applies, i.e., for all $\gamma<\gamma_1$, \eqref{eq:terminal_invariance} and \eqref{eq:terminal_cost_decrease} are satisfied.\\
Based on \eqref{eq:g_min} and \eqref{eq:g_max}, we introduce $c_{g_\mathrm{min}},\,c_{g_\mathrm{max}}\in\mathbb{R}$:
\begin{equation}
    c_{g_\mathrm{min}} \triangleq L_{g_\mathrm{min}}/\sqrt{\min\{\lambda_\mathrm{min}(Q), \lambda_\mathrm{min}(R)\}},
\end{equation}
with $c_{g_\mathrm{max}}$ defined similarly. Then, the following additional constants are defined for the input constraints:
\begin{equation}\label{eq:U_c}
    c_\mathrm{u}^i \triangleq \|P^{-1/2} K^\top H_{\mathrm{u},i}\|_2, \quad \forall i\in\mathbb{I}_1^{l_\mathrm{u}},
\end{equation}
where $H_{\mathrm{u},i}$ corresponds to the $i$\textsuperscript{th} half-space constraint in $U$. Note that \eqref{eq:U_c} is based on the support function of the ellipsoid defining $\mathcal{Z}_\mathrm{f}$~\cite{conte_distributed_2016}. To ensure that the constraints \eqref{eq:terminal_U_constr}--\eqref{eq:terminal_P_constr} are satisfied within the terminal region under $\kappa_\mathrm{f}$, the following notation is introduced:
\begin{subequations}
\begin{align*}
    &\gamma_{g_\mathrm{min}} \triangleq  g_\mathrm{min}^2(z_\mathrm{r}, u_\mathrm{r})/c_{g_\mathrm{min}}^2, \quad \gamma_{g_\mathrm{max}} \triangleq  g_\mathrm{max}^2(z_\mathrm{r}, u_\mathrm{r})/c_{g_\mathrm{max}}^2,\\
    &\gamma_\mathrm{u}^i \triangleq \left(H_{\mathrm{u},i}u_\mathrm{r} - h_{\mathrm{u},i}\right)^2/ ({c_\mathrm{u}^i})^2, \quad i\in\mathbb{I}_1^{l_\mathrm{u}}.
\end{align*}
\end{subequations}
Based on \cite[Lem.~3.37]{kohler_analysis_2021},
\begin{equation}\label{eq:gmin_ineq}
    g_\mathrm{min}(z, \kappa_\mathrm{f}(z)) \leq g_\mathrm{min}(z_\mathrm{r}, u_\mathrm{r}) + c_{g_\mathrm{min}} \sqrt{V_\mathrm{f}(z)}.
\end{equation}
To ensure that the chance constraint $g_\mathrm{min}$ is satisfied with the local control policy within $\mathcal{Z}_\mathrm{f}$, the right-hand side of \eqref{eq:gmin_ineq} should be $\leq 0$. Utilizing that $V_\mathrm{f}(z)\leq \gamma$ (on $z\in\mathcal{Z}_\mathrm{f}$), selecting $\gamma\leq\gamma_{g_\mathrm{min}}$ exactly ensures this condition. The same process can be repeated for $g_\mathrm{max}$ and $g_\mathrm{u}^i \triangleq H_{\mathrm{u},i}u - h_{\mathrm{u},i}$, $\forall i\in\mathbb{I}_1^{l_\mathrm{u}}$. Therefore, all considered constraints are satisfied inside the terminal region $\mathcal{Z}_\mathrm{f}=\{z:V_\mathrm{f}(z)\leq \mathrm{min}\{\gamma_1,\gamma_{g_\mathrm{min}}, \gamma_{g_\mathrm{max}}, \gamma_\mathrm{u}^1,\dots \gamma_\mathrm{u}^{l_\mathrm{u}}\}\}$ under the local control law $\kappa_\mathrm{f}$. Note that \eqref{eq:gmin_ineq} holds only if the terminal cost inequality \eqref{eq:terminal_cost_decrease} is fulfilled, but that is trivially satisfied for all $\gamma<\gamma_1$.\placeqed

\section{Proof of Lemma~\ref{lem:offset_cost}}\label{sec:appendix:offset_cost_proof}
Let $c^\ast\in C(r,\mathbb{S}_\mathrm{o})$ denote the shortest curve on $\mathbb{S}$ from $r$ to $\mathbb{S}_\mathrm{o}$, parametrized by arc-length $\rho \in [0, d_{\mathbb{S}_\mathrm{o}}(r)]$, with $c^\ast(0)=r$ and $c^\ast(d_{\mathbb{S}_\mathrm{o}}(r))=r_\mathrm{o} \in \mathbb{S}_\mathrm{o}$. Since $c^\ast$ is parametrized by arc-length, $\|{c^\ast}'(\rho)\|_2=1$. Select $\hat{r}=c^\ast(\epsilon d_{\mathbb{S}_\mathrm{o}}(r))\in\mathbb{S}$. Using that the straight line between $r$ and $\hat{r}$ is shorter than any curve on $\mathbb{S}$ connecting them: $\|\hat{r}-r\|_2 \leq \epsilon d_{\mathbb{S}_\mathrm{o}}(r)$, which satisfies \eqref{eq:curvature_assum_1} with $k_0=1$. Define the cost along the curve as $\phi(\rho) = V_\mathrm{r}(c^\ast(\rho))$. Since $V_\mathrm{r}$ is $\mathcal{C}_2$, $c^\ast$ is twice differentiable on the $\mathcal{C}_2$ manifold $\mathbb{S}$, and $\breve{\mathbb{S}}$ is compact, the second derivative of $\phi$ w.r.t. the arc-length $\rho$, i.e., $\phi''(\rho)$, is well-defined and continuous. Let $\mathbb{K}\subseteq\mathbb{S}$ be a fixed compact neighborhood of $r_\mathrm{c}$. Since $\mathbb{S}$ is $\mathcal{C}_2$, the curvature of $c^\ast$ in $\mathbb{K}$ is uniformly bounded, i.e., $\|{c^\ast}''(\rho)\|_2\leq\kappa_{\max}<\infty$, with $\kappa_\mathrm{max}$ independent of $\lambda$.
As $r_\mathrm{c}\in\mathbb{S}$, any $r_\mathrm{o}\in\mathbb{S}_\mathrm{o}$ satisfies $V_\mathrm{r}(r_\mathrm{o})\leq V_\mathrm{r}(r_\mathrm{c})$; then, due to the non-negativity of the KL divergence, $\|r_\mathrm{o}-r_\mathrm{c}\|_2\leq\sqrt{V_\mathrm{r}(r_\mathrm{c})/\lambda}$, where $V_\mathrm{r}(r_\mathrm{c})$ is independent of $\lambda$. Let us choose $\breve{\mathbb{S}}=\{r\in\mathbb{S}\mid
d_{\mathbb{S}_\mathrm{o}}(r)\leq\bar{d}\}$. As $\|r-r_\mathrm{c}\|_2\leq\|r-r_\mathrm{o}\|_2+\|r_\mathrm{o}-r_\mathrm{c}\|_2$, it follows that there exist a sufficiently small $\bar{d}$ and a sufficiently large $\lambda$
such that $\breve{\mathbb{S}}\subseteq\mathbb{K}$ and
$\delta\triangleq\sup_{r\in\breve{\mathbb{S}}}\|r-r_\mathrm{c}\|_2\leq
\bar{d}+\sqrt{V_\mathrm{r}(r_\mathrm{c})/\lambda}<1/\kappa_{\max}$. Note that $\breve{\mathbb{S}}$ is compact and the shortest curves starting in $\breve{\mathbb{S}}$ remain there. Then,
\begin{multline}\label{eq:pen_hessian}
    \frac{\mathrm{d}^2}{\mathrm{d}\rho^2}\lambda\|c^\ast(\rho)-r_\mathrm{c}\|_2^2
    = 2\lambda\|{c^\ast}'(\rho)\|_2^2
    + 2\lambda\,(c^\ast(\rho)-r_\mathrm{c})^\top{c^\ast}''(\rho)\\
    \geq 2\lambda\big(1-\|c^\ast(\rho)-r_\mathrm{c}\|_2\,\kappa_{\max}\big)
    \geq 2\lambda(1-\delta\kappa_{\max})>0.
\end{multline}
Denoting by $l^-\geq 0$ the largest negative part of the second derivative of the KL term along $c^\ast$ on $\mathbb{K}$, which is finite and independent of $\lambda$ by the $\mathcal{C}_2$ property and compactness, we obtain $\phi''(\rho)\geq 2\lambda(1-\delta\kappa_{\max})-l^-\triangleq\mu>0$ for a sufficiently large $\lambda$. Because $r_\mathrm{o}$ is a minimizer of $V_\mathrm{r}$ on $\mathbb{S}$, the derivative along the manifold vanishes at the end-point when $r_\mathrm{o}\in\mathrm{int}(\mathbb{S})$, i.e., $\phi'(d_{\mathbb{S}_\mathrm{o}}(r)) = 0$. Otherwise, if $r_\mathrm{o}$ lies on the boundary of $\mathbb{S}$, then the first-order optimality condition gives $\phi'(d_{\mathbb{S}_\mathrm{o}}(r))\leq 0$, i.e., that the gradient along $c^\ast$ points inward (or zero), ensuring that no descent directions exist within $\mathbb{S}$. By the Mean Value Theorem applied to $\phi'$, for any $r\in\breve{\mathbb{S}}$ and $\rho\in[0, d_{\mathbb{S}_\mathrm{o}}(r)]$, there exists a point $s\in[\rho, d_{\mathbb{S}_\mathrm{o}}(r)]$ for which: $\phi'(d_{\mathbb{S}_\mathrm{o}}(r)) - \phi'(\rho) = \phi''(s)(d_{\mathbb{S}_\mathrm{o}}(r) - \rho)$. Substituting $\phi'(d_{\mathbb{S}_\mathrm{o}}(r)) \leq 0$ and $\phi''(\rho)\geq \mu$, we obtain:
\begin{equation}\label{eq:mvt_result}
\phi'(\rho) \leq -\mu(d_{\mathbb{S}_\mathrm{o}}(r) - \rho).
\end{equation}
Applying the fundamental theorem of calculus to $\phi(\rho)$:
\begin{equation}
    V_\mathrm{r}(\hat{r}) - V_\mathrm{r}(r) = \phi(\epsilon d_{\mathbb{S}_\mathrm{o}}(r)) - \phi(0) = \int_0^{\epsilon d_{\mathbb{S}_\mathrm{o}}(r)} \phi'(\rho)~\mathrm{d}\rho.
\end{equation}
Substituting the gradient bound \eqref{eq:mvt_result}:
\begin{multline}
V_\mathrm{r}(\hat{r}) - V_\mathrm{r}(r) \leq \int_0^{\epsilon d_{\mathbb{S}_\mathrm{o}}(r)} -\mu (d_{\mathbb{S}_\mathrm{o}}(r) - \rho)~\mathrm{d}\rho = \\
-\mu \left( 1 - \frac{\epsilon}{2} \right) \epsilon d_{\mathbb{S}_\mathrm{o}}(r)^2.
\end{multline}
Since $\epsilon \in [0, 1]$, we have $(1 - \epsilon/2) \geq 1/2$. Thus, setting $k_1 = \mu /2 $ satisfies \eqref{eq:curvature_assum_2}.\placeqed

\end{document}